\documentclass[10pt,a4paper]{article}

\usepackage[english]{babel}
\usepackage{xcolor}
\usepackage[a4paper,top=2cm,bottom=2cm,left=3cm,right=3cm,marginparwidth=1.75cm]{geometry}

\usepackage{graphicx}
\usepackage[colorlinks=true, allcolors=blue]{hyperref}
\usepackage{authblk}

\usepackage{amsmath}
\usepackage{amsthm}
\usepackage{amsfonts}
\usepackage{enumerate}
\usepackage{bm}
\usepackage[T1]{fontenc}
\usepackage{todonotes}

\theoremstyle{plain}
\newtheorem{theorem}{Theorem}[section]
\newtheorem{proposition}[theorem]{Proposition}
\newtheorem{lemma}[theorem]{Lemma}
\newtheorem{corollary}[theorem]{Corollary}
\newtheorem{remark}[theorem]{Remark}

\theoremstyle{definition}
\newtheorem{definition}[theorem]{Definition}
\newtheorem{example}[theorem]{Example}

\newcommand{\RR}{\mathbb R}

\newcommand{\bM}{\mathbf{M}}
\newcommand{\bF}{\mathbf{F}}
\newcommand{\bB}{\mathbf{B}}
\newcommand{\be}{\mathbf{e}}
\newcommand{\dd}{\text{d}}

\title{Towards an effective medium theory for in-vivo Magnetic Resonance: Characterizing Admissible Magnetization Dynamics}
\author[1]{Alessandro Sbrizzi\thanks{Alessandro Sbrizzi and Ray Sheombarsing contributed equally to this work. Corresponding author: a.sbrizzi@umcutrecht.nl}}
\author[1]{Miha Fuderer}
\author[1]{Ray Sheombarsing$^{\ast}$}

\affil[1]{Utrecht University Medical Center}
\begin{document}
\maketitle

\begin{abstract}
The classical Bloch equation forms the foundation of magnetic resonance theory but may be insufficient to capture the complex effective dynamics observed in heterogeneous materials such as biological tissue. Here an effective-medium strategy is proposed to describe magnetic resonance dynamics from observations at macro-scale (order 1 mm$^3$). In this framework, the classical Larmor torque is retained while the relaxation term is replaced by a general nonlinear field (i.e. the effective medium term). This field is required to preserve key physical properties, including rotational symmetry, global stability of thermal equilibrium, and dissipativity. These requirements lead to a complete characterization of admissible relaxation fields and provide a constructive framework for generating nonlinear relaxation models. The resulting theory recovers classical Bloch dynamics as a special case and provides a principled foundation for effective-medium descriptions of magnetic resonance phenomena.
\end{abstract}

\section{Introduction}

The Bloch equation is the foundational dynamical model of Nuclear Magnetic Resonance (NMR) and Magnetic Resonance Imaging (MRI) \cite{bloch1946nuclear}. It accurately describes MR signals in simple homogeneous materials and provides the basis for the design, analysis, and interpretation of MRI experiments 
\cite{brown2014magnetic}. Its impact has been so profound that the Bloch framework remains the standard language through which MR phenomena are understood nearly eighty years after its introduction.
However, modern MRI is primarily applied to biological tissue, a highly heterogeneous material whose microscopic complexity spans multiple spatial and temporal scales. Each MRI voxel contains a large ensemble of spins distributed across cellular, extracellular, and vascular environments, all interacting through a variety of physical and biological processes. While MRI measures an aggregate signal at the tissue scale, the original Bloch equation was not formulated to explicitly account for such complexity. Consequently, although it remains highly successful for describing basic MR dynamics, its direct application to quantitative characterization of tissue presents important challenges.

One manifestation of these limitations is observed in quantitative MRI (qMRI), and particularly in relaxometry. In vivo measurements of parameters such as $T_1$ and $T_2$ often exhibit protocol dependence, meaning that estimated parameter values may vary across acquisition protocols even when probing the same tissue \cite{stikov2015accuracy,bojorquez2017normal,pai2008comparative,matzat2015t2,wright2008water}. Part of this variability can be attributed to experimental factors such as off-resonance effects, imperfect B1 calibration, and eddy currents \cite{leitao2021efficiency,sled2000correction,zhao2020impact,cercignani2018quantitative}. Nevertheless, there is broad agreement that the original Bloch equation, in its 1946 form, cannot fully capture the intrinsic complexity of living tissue \cite{weiskopf2021quantitative}. This observation has motivated extensive efforts to extend the Bloch framework in order to describe increasingly complex MR dynamics.
Notable examples include multi-compartment models \cite{kroeker1986analysis}, the Bloch-Torrey equation \cite{torrey1956bloch}, which incorporates diffusion, and the Bloch-McConnell equation \cite{mcconnell1958reaction,henkelman1993quantitative}, which describes magnetization exchange between distinct molecular environments. These and related developments were originally introduced to explain NMR observations in molecular systems, well before MRI became a tool for in vivo tissue characterization. Their subsequent application to biological tissue has substantially advanced our understanding of MR signal formation and has enabled many important imaging techniques \cite{le2003looking,grossman1994magnetization}. At the same time, these extensions have not completely resolved the problem of obtaining protocol-independent tissue parameters in vivo \cite{keenan2022challenges}. More generally, they pursue a common strategy: explaining increasingly complex observations by introducing additional microscopic components, compartments, pools, or exchange pathways. While powerful, this bottom-up approach naturally leads to models of increasing dimensionality and parameterization, often requiring model simplifications or reliance on literature values when applied to tissue \cite{henkelman1993quantitative,ag2019fast,sled2018modelling,malik2018extended,asslander2025magnetization}.

Recognizing the significant advances made possible by these developments, we propose a complementary perspective. The inspiration comes from the Bloch equation itself and from the philosophy underlying its original derivation. The relaxation terms introduced by Bloch were not derived from a detailed microscopic description; rather, they were proposed phenomenologically to account for experimentally observed behavior. We suggest that a similar approach may be fruitful at the tissue scale (i.e. spatial resolution in the order of mm$^3$). Instead of seeking increasingly detailed microscopic descriptions of tissue organization, one may attempt to formulate an effective theory of \emph{tissue-scale} MR dynamics that directly captures the behavior observed by MRI. Ideally, such a theory would yield model parameters that are consistent across acquisition protocols, or at least reduce the variability that currently affects quantitative measurements.

In this paper, we introduce a theoretical framework intended to guide the development of such a phenomenological description of tissue-scale MR behavior. The framework remains grounded in Bloch's original approach while generalizing it to accommodate the complexity of biological materials. Its emphasis is on describing observable dynamics at the macroscopic scale of MRI, rather than increasing the number of microscopic constituents represented in the model. In this sense, our goal is to develop an effective-medium theory of tissue-scale MR dynamics \cite{choy2015effective,noid2013perspective}.
Specifically, we consider a class of models in which the dynamics are governed by an ordinary differential equation composed of the established rotational term (Larmor torque)  together with a more general, potentially nonlinear dissipative term. This dissipative component is intended to capture tissue-dependent effects, including relaxation, exchange, micro-structural organization, and other mechanisms that influence the observed signal. Importantly, the model retains the same dimensionality as the Bloch equation, describing the evolution of a three-dimensional magnetization vector whose transverse components remain directly observable. In this way, we avoid the progressive expansion of the state space that characterizes many multi-compartment and exchange-based models. The order reduction is motivated by earlier efforts to simplify the description of MR dynamics, particularly in the context of Magnetization Transfer (MT) modelling. For example, Graham and Henkelman \cite{graham1997understanding} reduced the six-dimensional two-pool Bloch system to an effective four-dimensional model by neglecting the transverse dynamics of the macromolecular pool. Subsequently, Teixeira et al. \cite{ag2019fast} further reduced the steady-state dynamics of the resulting four-dimensional two-pool system to an effective three-dimensional single-pool description. In their approach, the free-water dynamics are described by the conventional Bloch equations with effective parameters for thermal equilibrium state and longitudinal relaxation rate, whose values depend on the properties of both pools as well as on the applied RF irradiation. While these reduced-order models have proven effective within their respective domains of validity, their reliance on problem-specific assumptions motivates the search for a more general mathematical framework for describing magnetic resonance dynamics. The present work pursues this objective by deriving a broader class of effective dynamical models that can encompass a variety of physical mechanisms and modeling assumptions.

A fundamental requirement of the proposed framework is that the resulting dynamics reduce to the standard Bloch equation in the limit of simple homogeneous materials. To achieve this goal, we introduce a set of postulates that characterize Bloch dynamics within existing NMR theory and derive admissibility conditions on the functional form of the generalized dissipative term. Based on these admissibility criteria, we further derive general strategies for constructing candidate tissue-scale MR models and provide some  concrete examples illustrating the richer dynamics captured by the proposed generalization. The tissue-scale MR models can subsequently be confronted with experimental MRI data. The resulting framework establishes the foundations of a generalized Bloch theory, intended as a phenomenological description of biological tissue at the scale of MRI measurements.

\paragraph{Overview} The remainder of this paper is organized as follows. In
Section~\ref{sec:problem_setup}, we review the classical Bloch relaxation model and discuss the
physical principles that a generalized relaxation mechanism should satisfy. In
Section~\ref{sec:formal_characterization}, these principles are formulated in precise
mathematical terms, leading to the notion of an admissible relaxation field. We then derive a
complete characterization of all admissible relaxation fields and establish the corresponding
structure theorem. Finally, we present a constructive procedure that yields a large and flexible
family of admissible nonlinear relaxation models and provide some examples.

\section{Problem setup} 
\label{sec:problem_setup}
In this section we provide a formal description of what we mean by finding a physically admissible
generalization of the relaxation term in the Bloch equation. Throughout, $(\be_1,\be_2,\be_3)$
denotes the standard orthonormal basis of $\RR^3$, and $\langle \cdot,\cdot \rangle$ denotes the
Euclidean inner product.

\subsection{Bloch equation and relaxation dynamics}
In an idealized MRI experiment, the magnetic field $\bB$ consists of a strong homogeneous static
field in the longitudinal direction, an RF excitation field in the transverse plane, and magnetic
field gradients used for spatial encoding. These fields interact with the time-dependent macroscopic 
magnetization vector $\bM = \begin{pmatrix} M_{1} & M_{2} & M_{3} \end{pmatrix}^T
: [0,\infty) \to \RR^3$ according to the Larmor torque dynamics:
\begin{align*}
	\dot \bM(t) = \gamma \bM(t) \times \bB(t), \quad t > 0. 
\end{align*}
Here $\gamma >0$ is the gyromagnetic ratio of the nuclei under consideration. To describe observed
NMR behavior in laboratory tests, Bloch added relaxation effects and proposed his celebrated
equation:
\begin{align}
	\label{eq:bloch_eqs_lab_frame}
	\dot \bM(t) = \gamma \bM(t) \times \bB(t) - \begin{pmatrix}
		\displaystyle T_{2}^{-1} & 0 & 0 \\
		0 & \displaystyle T_{2}^{-1}& 0 \\
		0 & 0 & \displaystyle T_{1}^{-1}
	\end{pmatrix} \left( \bM - \bM^{\ast} \right), 
	\quad t > 0. 
\end{align}
The second term describes the return of the magnetization toward the thermal equilibrium state
$\bM^{\ast} = M^{\ast} \be_{3}$, where $M^{\ast}>0$, and $T_1,T_2>0$ are the longitudinal and
transverse relaxation times, respectively. 

The goal of this paper is to develop models for the MR dynamics that generalize
\eqref{eq:bloch_eqs_lab_frame}. We keep the established rotational part (Larmor torque) as it is and
focus on the phenomenological part, the relaxation term. To better study the latter we isolate it by
assuming the RF excitation and gradient fields have been switched off. Since we are interested only
in post-excitation relaxation, we work in the usual rotating frame. In this frame, the contribution
of the static magnetic field does not appear in the ideal on-resonance dynamics. Consequently, after
the RF excitation field and magnetic field gradients have been switched off, the magnetization $\bM$  evolves according to
\begin{align}
	\label{eq:bloch_eqs_rot_frame}
	\begin{cases}
		\dot{\bM}(t) = \bF \left( \bM(t) \right), & t > 0, \\
		\bM(0) = \bM_{0},
	\end{cases}
\end{align}
where $\bF: \RR^{3} \rightarrow \RR^{3}$ is the relaxation vector field given by
\begin{align}
    \label{eq:affine_relaxation_field}
	\bF(\bM) = - \begin{pmatrix}
		\displaystyle T_{2}^{-1} & 0 & 0 \\
		0 & \displaystyle T_{2}^{-1}& 0 \\
		0 & 0 & \displaystyle T_{1}^{-1}
	\end{pmatrix} \left( \bM - \bM^{\ast} \right).
\end{align}
Here we have abused notation and denoted the magnetization in the rotating frame by $\bM$ as well.
The initial state $\bM_0\in\RR^3$ represents the magnetization at the instant that the RF excitation
and encoding gradients are switched off.

In this paper we replace the affine relaxation vector field in \eqref{eq:affine_relaxation_field} 
with a general class of nonlinear vector fields $\bF$ and interpret $\bM$ as the vector whose transverse 
components describe the macroscopically observed magnetization. Note that $\bM$ is thus an effective 
medium description of all microscopic contributions to the MR signal and may encompass multiple molecular 
environments, exchange effects, microscopic motion and other sub-voxel phenomena. We identify nonlinear 
relaxation models for $\bM$ that retain the essential physical features of classical Bloch relaxation 
(see Section \ref{sec:physical_constraints}) while allowing for more general relaxation behavior. 

\subsection{Physical constraints}
\label{sec:physical_constraints}
Since we aim at principles that can guide
the search for phenomenological descriptions of $\bF$, we set out to determine physically plausible
constraints for it. We first discuss these constraints informally and explain their physical
motivation. Afterwards, we provide a precise mathematical formulation.

The requirements we impose for physical admissibility are:
\begin{enumerate}
	\item{(Equivariance)} $\bF$ is equivariant with respect to rotations around $\be_{3}$.
	\item{(Thermal equilibrium)} $\bM^{\ast}$ is the unique equilibrium of $\bF$. 
	\item{(Global asymptotic stability)} The magnetization relaxes to $\bM^{\ast}$ 
        for any initial condition $\bM_0$. 
    \item{(Dissipativity)} The dynamics admit an ``energy'' function that acts as a Lyapunov 
        potential and is non-increasing along trajectories.
\end{enumerate}
The first requirement reflects the axial symmetry of the relaxation process about the longitudinal
direction $\be_{3}$. Once the RF excitation and gradient fields have been switched off, no preferred
direction remains in the transverse plane $\text{span}\{\be_{1}, \be_{2} \}$.  Consequently,
magnetization states that differ only by a rotation about $\be_{3}$ should exhibit identical
relaxation behavior, up to the same rotation.

Requirements $(2)$ and $(3)$ reflect the experimentally observed tendency of the magnetization to
recover toward the thermal equilibrium state $\bM^{\ast}$. In the classical Bloch model, these
properties are reflected in two characteristic relaxation mechanisms: decay of the transverse
magnetization and recovery of the longitudinal component. The requirement that $\bM^{\ast}$ is
the unique equilibrium, and is globally asymptotically stable, guarantees that every integral 
curve of $\bF$ converges to the thermal equilibrium state, irrespective of its initial condition. 

The final requirement $(4)$ reflects the absence of any external driving mechanism once the RF
excitation and gradient fields have been switched off. At that stage, no external source remains
that can drive the magnetization further away from thermal equilibrium, and the subsequent evolution
should therefore be purely dissipative. To express this mathematically, we require the (squared)
Euclidean distance $V(x) = \frac{1}{2}\Vert x - \bM^{\ast} \Vert_{2}^{2}$ to the equilibrium 
$\bM^{\ast}$ to strictly decrease along every non-equilibrium integral curve. Roughly speaking,
$V$ may be interpreted as a potential energy-like measure of deviation from the equilibrium. This
condition is substantially stronger than requirement (4) and global asymptotic stability, as it 
rules out even temporary excursions away from equilibrium. As we shall see later, the dissipativity 
condition is very closely related to requirements $(2)$ and $(3)$.

Taken together, conditions $(1)$--$(4)$ capture the symmetry, stability, and dissipation properties
that we regard as essential for physically admissible generalizations of the Bloch relaxation term.
As it turns out, these seemingly modest requirements place strong restrictions on the structure of
admissible relaxation vector fields.

\section{Formal characterization of admissible relaxation fields}
\label{sec:formal_characterization}
In this section we provide a formal characterization of physically admissible relaxation fields
satisfying conditions $(1)$--$(4)$. We derive a complete and explicit description of all such vector
fields and conclude with a constructive recipe for generating them. 

\paragraph{Notation} 
Throughout, we consider vector fields $\bF:\RR^3\to\RR^3$ whose integral curves describe the time
evolution (relaxation) of the magnetization. To distinguish integral curves from generic points in phase space, we
write $x=(x_1,x_2,x_3)\in\RR^3$ for points and $\bM:[0,T]\to\RR^3$ for integral curves of $\bF$. Furthermore, 
we will frequently encounter functions $G: \RR^{3} \rightarrow \RR$ of the form 
\begin{align*} 
    G(x)=g(r,x_3), \qquad
    r=r(x_1,x_2)=\sqrt{x_1^2+x_2^2},
\end{align*} 
where $g:[0,\infty)\times\RR\to\RR$ is given. When discussing properties of $g$, we denote its first argument
by $s$ and write $g(s,x_3)$, reserving $r$ for the radial coordinate.

\subsection{Admissible relaxation fields}
In order to provide a full characterization of physically plausible relaxation vector fields, we
first formulate the physical requirements of Section~\ref{sec:physical_constraints} in precise 
mathematical terms. To this end, let $R_{3}(\theta) \in \mathrm{SO}(3)$ denote the rotation operator 
around $\be_{3}$ with angle $\theta$, i.e.,
\begin{align*}
	R_{3}(\theta) = 
	\begin{pmatrix}
		\cos(\theta) & - \sin(\theta) & 0 \\
		\sin(\theta) &  \cos(\theta) & 0 \\
		0 & 0 & 1
	\end{pmatrix}.
\end{align*} 
\begin{definition}[Admissible relaxation fields]
	\label{def:admissible_relaxation_fields}
    A continuously differentiable vector field $\bF: \RR^{3} \rightarrow \RR^{3}$ is called an
    admissible relaxation field if the following conditions hold: 
	\begin{enumerate}
		\item[(\textbf{C1})] $\bF( R_{3}(\theta) x ) = R_{3}(\theta) \bF(x)$ for all $x \in \RR^{3}$ and $\theta \in [0, 2 \pi)$. 
		\item[\textbf{(C2)}] $\bF(\bM^{\ast}) = 0$ and $\bF(x) \not = 0$ for all $x \in \RR^{3} \setminus \{ \bM^{\ast} \}$. 
		\item[\textbf{(C3)}] $\sigma \left( D\bF(\bM^{\ast}) \right) \subset (-\infty, 0)$, where $D\bF$ denotes the derivative 
                             of $\bF$ and $\sigma$ its spectrum. 
		\item[\textbf{(C4)}] Every non-equilibrium integral curve $\bM:[0, \tau(\bM_{0})) \rightarrow \RR^{3}$ of $\bF$ satisfies 
                            \begin{align*}
                                \dfrac{\dd}{\dd t} \left( t \mapsto V(\bM(t)) \right) < 0, \quad \forall t \in [0,  \tau(\bM_{0})),
                            \end{align*} 
                            where $\tau(\bM_{0}) >0$ is the maximal integration time associated to the initial condition 
                            $\bM_{0} \in \RR^{3} \setminus \{ \bM^{\ast} \}$, and $V: \RR^{3} \rightarrow [0, \infty)$ is defined 
                            by $V(x) := \frac{1}{2} \Vert x - \bM^{\ast} \Vert_{2}^{2}$. 
		\end{enumerate}
\end{definition}

The relationship between Definition~\ref{def:admissible_relaxation_fields} and the physical
requirements of Section~\ref{sec:physical_constraints} is straightforward. Conditions \textbf{(C1)},
\textbf{(C2)}, and \textbf{(C4)} formalize requirements $(1)$, $(2)$, and $(4)$, respectively.
Moreover, \textbf{(C4)} implies that solutions exist for all positive times, as trajectories remain
confined to bounded sublevel sets of $V$, thereby excluding finite-time blow-up. This is of course a
minimal requirement  for physical consistency, which we left implicit until now. Furthermore,
conditions \textbf{(C3)} and \textbf{(C4)} imply that the unique equilibrium $\bM^\ast$ is a globally
asymptotically stable hyperbolic sink with real negative eigenvalues. 

Condition \textbf{(C1)} imposes a strong restriction on the linearization at the equilibrium. In
particular, in combination with $\textbf{(C3)}$, it forces $D\bF(\bM^\ast)$ to have the same
structure as the linear relaxation operator of the classical Bloch model. As a result, the
\emph{linearized} dynamics around $\bM^{\ast}$ are characterized by two distinguished decay 
rates: one governing relaxation in the transverse plane and one governing relaxation along the 
longitudinal direction.

A formal statement and proof of the above assertions, and in particular that admissible relaxation
fields satisfy the desired physical constraints $(1)$--$(4)$, can be found in Appendix \ref{sec:appendix}.

\begin{remark}[Strict Lyapunov function]
    Condition (\textbf{C4}) is equivalent to the statement that $V$ is strictly decreasing along 
    non-equilibrium integral curves of $\bF$. In particular, this implies that $V$ is a so-called strict 
    Lyapunov function for $\bM^{\ast}$, since $V( \bM^{\ast}) = 0$, $V(x) > 0$ for $x \in \RR^{3} \setminus 
    \{ \bM^{\ast} \}$, and $V$ is radially unbounded. Such functions may be viewed as a 
    generalization of the notion of a potential and are a characteristic feature of gradient-like 
    dynamical systems. The proof of global asymptotic stability relies heavily on the existence of 
    this strict Lyapunov function, see Appendix \ref{appendix:global_asymptotic_stability}.
\end{remark}

\subsection{Characterization of admissible relaxation fields}
We now provide a complete characterization of admissible relaxation fields. Effectively,
equivariance with respect to rotations about $\be_{3}$ forces the vector field to have a highly
structured form, which in turn enables an easy characterization of conditions
$(\textbf{C2})$--$(\textbf{C4})$. Although rotational symmetry is most naturally analyzed in
cylindrical coordinates, such coordinates are singular on the longitudinal axis, that is, on 
$\text{span}\{ \be_{3} \}$. Since the equilibrium state $\bM^{\ast}$ lies on this axis, 
conditions \textbf{(C2)}--\textbf{(C4)} are most conveniently studied in Cartesian coordinates. 
For this reason, we formulate the characterization theorem in Cartesian coordinates.
\begin{theorem}[Characterization of admissible relaxation fields]
	\label{thm:admissible_relaxation_fields}
    Let $\bF: \RR^{3} \rightarrow \RR^{3}$ be a continuously differentiable vector field. Then $\bF$
    is an admissible relaxation field if and only if there exists $\alpha, \beta \in C^{1}\left( (0,
    \infty) \times \RR \right) \cap C([0, \infty) \times \RR)$, $\zeta \in C^{1}([0, \infty) \times
    \RR)$ such that
	\begin{align}
		\label{eq:admissible_relaxation_field}
		\bF(x) = 
		\begin{pmatrix}
			\alpha(r,x_{3}) & -\beta(r, x_{3}) & 0 \\
			\beta(r, x_{3}) & \alpha(r, x_{3}) & 0 \\
			0 & 0 & 0 
		\end{pmatrix}x + 
		\begin{pmatrix}
			0 \\ 0 \\ \zeta(r, x_{3})
		\end{pmatrix},
	\end{align}	
	where $r = r(x_{1}, x_{2}) = \sqrt{x_{1}^{2} + x_{2}^{2}}$, and
	\begin{itemize}
	    \item[$(i)$] For all $\tilde x_{3} \in \RR$, we have 
            \begin{align*}
                \lim_{(s, x_{3}) \rightarrow (0, \tilde x_{3})} s \partial_{s} \alpha(s, x_{3}) &= 0, &
                \lim_{(s, x_{3}) \rightarrow (0, \tilde x_{3})}  s \partial_{x_{3}} \alpha(s, x_{3}) &= 0, & \\[2ex]
                \lim_{(s, x_{3}) \rightarrow (0, \tilde x_{3})}  s \partial_{s} \beta(s, x_{3}) &= 0, &
                \lim_{(s, x_{3}) \rightarrow (0, \tilde x_{3})} s \partial_{x_{3}} \beta(s, x_{3})&= 0,
            \end{align*}
            and moreover, $\partial_{s} \zeta(0, \cdot) \equiv 0$.
        \item[$(ii)$]  The only zero of $\zeta(0, \cdot)$ is $M^{\ast}$.
        \item[$(iii)$] $\alpha(0,M^{\ast}) < 0$, $\beta(0, M^{\ast})=0$ and
            $\partial_{x_{3}}\zeta(0,M^{\ast}) < 0$.
        \item[$(iv)$] $s^{2} \alpha(s,x_{3}) + (x_{3}-M^{\ast}) \zeta(s, x_{3}) < 0$, $\forall (s,
            x_{3}) \in \left( [0, \infty) \times \RR \right) \setminus \{(0, M^{\ast})\}$. 
	\end{itemize}
\end{theorem}
\begin{remark}[Notation]
    We shall frequently denote the vector field in \eqref{eq:admissible_relaxation_field} by
    $\bF_{(\alpha, \beta, \zeta)}$. 
\end{remark}
We refer the reader to Appendix \ref{appendix:proofs_main_theorems} for a proof and restrict
ourselves here to a brief discussion of the result. The equivariance condition \textbf{(C1)} is
equivalent to the requirement that $\bF$ admits the representation in 
\eqref{eq:admissible_relaxation_field} with components $(\alpha, \beta, \zeta)$ satisfying $(i)$.
The conditions in $(i)$ are precisely those needed to ensure that the resulting vector field extends 
smoothly to the symmetry axis $\text{span}\{\be_{3}\}$. Furthermore, assumptions $(ii)$--$(iv)$ 
together encode conditions \textbf{(C2)}--\textbf{(C4)}. 
\begin{example}[Classical Bloch model]
	The classical Bloch relaxation term is recovered as the special case
	\begin{align*}
		\alpha \equiv - \frac{1}{T_{2}}, \quad
		\beta(s, x_{3}) \equiv 0, \quad
		\zeta(s, x_{3}) = -\frac{x_{3} - M^{\ast}}{T_{1}},
	\end{align*}
    which clearly satisfies the assumptions of Theorem \ref{thm:admissible_relaxation_fields}. 
\end{example}

The interpretation of the functions defining $\bF$ becomes particularly transparent in cylindrical
coordinates. The function $\alpha$ governs the transverse radial dynamics, $\beta$ governs
rotational motion around the longitudinal axis $\be_{3}$, and $\zeta$ determines the longitudinal
dynamics. This is most clearly seen in cylindrical coordinates, in which
\eqref{eq:admissible_relaxation_field} takes the form
\begin{align*}
	\bF_{(\alpha, \beta, \zeta)}(r, \theta, x_{3}) \equiv
	\bF_{(\alpha, \beta, \zeta)}(r, x_{3}) =
	\begin{pmatrix}
		r \alpha(r, x_{3}) \\
		\beta(r, x_{3}) \\
		\zeta(r, x_{3})
	\end{pmatrix}, \quad 
	(r, \theta, x_{3}) \in (0, \infty) \times [0, 2 \pi) \times \RR. 
\end{align*}
Therefore, if $t \mapsto \left( M_{r}(t), M_{\theta}(t), M_{3}(t) \right)$ is the
cylindrical-coordinate representation of an integral curve of $\bF_{(\alpha, \beta, \zeta)}$, then 
\begin{align}
	\label{eq:cylindrical_coords}
	\begin{cases}
		\dot M_{r} = M_{r} \alpha(M_{r}, M_{3}), \\
		\dot M_{\theta} = \beta(M_{r}, M_{3}), \\
		\dot M_{3} = \zeta(M_{r}, M_{3}).
	\end{cases}
\end{align}
This representation describes the transverse radial dynamics, the rotational dynamics, and
the longitudinal dynamics more clearly than the Cartesian representation.

\subsection{Explicit construction of admissible relaxation fields}
While Theorem~\ref{thm:admissible_relaxation_fields} provides a complete characterization of
admissible relaxation fields, it does not directly yield a practical recipe for constructing
functions $(\alpha,\beta,\zeta)$ satisfying the required assumptions. In this section we propose
such a construction and obtain a computable parameterization of a large family of admissible
relaxation fields. As before, we only outline the main ideas and results; the complete details are
deferred to Appendices \ref{appendix:proofs_main_theorems} and 
\ref{appendix:construction_admissible_fields}.

The main idea behind the construction is to exploit the dissipativity condition $(iv)$, which
imposes strong structural constraints on the pair $(\alpha,\zeta)$. This leads to a natural family
of admissible candidates, after which conditions $(i)$--$(iii)$ can be enforced. Since condition
$(i)$ is a bit technical, and may obscure the main idea, we refer the reader to Appendix
\ref{appendix:construction_admissible_fields} for an explicit way to construct such mappings.
Here we focus on conditions $(ii)$--$(iv)$ and assume we know (see Lemma's 
\ref{lemma:construction_V1} and \ref{lemma:construction_V2}) how to construct mappings in 
\begin{align*}
    \mathcal{V}_{1}
    &:= \left \{ g \in C^{1}((0,\infty)\times\mathbb{R})
    \cap C([0,\infty)\times\mathbb{R}) \;\Bigg|\;
    \begin{aligned}
    &\lim_{(s,x_{3})\to(0,\tilde{x}_{3})}
      s\,\partial_s g(s,x_{3}) = 0,\\
    &\lim_{(s,x_{3})\to(0,\tilde{x}_{3})}
      s\,\partial_{x_{3}}g(s,x_{3}) = 0,
      \ \forall\,\tilde{x}_{3}\in\mathbb{R}
    \end{aligned}
    \right \}, \\[2ex]
    \mathcal{V}_{2} &:=
    \left \{
    	g \in C^{1}([0, \infty) \times \RR) \mid \partial_{s}g(0, \cdot) \equiv 0
    \right \}.
\end{align*}
With this notation in place, condition $(i)$ of Theorem \ref{thm:admissible_relaxation_fields} is
equivalent to $\alpha, \beta \in \mathcal{V}_{1}$ and $\zeta \in \mathcal{V}_{2}$. 
\begin{proposition}[Construction of admissible relaxation fields]
    \label{prop:construction_admissible_fields}
    Let $a, \beta, \psi \in \mathcal{V}_{1}$ and $c \in \mathcal{V}_{2}$. Assume 
    $a, c >0$, $\beta(0, M^{\ast})=0$. Then 
	\begin{align}
		\label{eq:construction}
		\bF(x) &= 
		\resizebox{0.85\linewidth}{!}{
		$\begin{pmatrix}
			- a(r, x_{3}) -  (x_{3}-M^{\ast})^{2} \psi(r, x_{3}) & -\beta(r, x_{3}) & 0 \\
			\beta(r, x_{3}) & -a(r, x_{3}) -  (x_{3}-M^{\ast})^{2} \psi(r, x_{3}) & 0 \\
			0 & 0 & r^{2} \psi(r, x_{3}) - c(r, x_{3})
		\end{pmatrix}$} x \nonumber \\[2ex] & \quad + 
		M^{\ast} \begin{pmatrix}
			0 \\ 0 \\ c(r, x_{3}) - r^{2} \psi(r, x_{3})
		\end{pmatrix}.
	\end{align}	
	is an admissible relaxation field. 	
\end{proposition}
Proposition~\ref{prop:construction_admissible_fields} provides a constructive sufficient condition
for admissibility. It does not parameterize all admissible relaxation fields, but it yields a large
and flexible class of examples that automatically satisfy the requirements of
Theorem~\ref{thm:admissible_relaxation_fields}. The constructed fields retain the qualitative
structure of the classical Bloch model. The functions $a$ and $c$ determine the dissipative part of
the dynamics, $\beta$ governs rotational motion about the longitudinal axis, and $\psi$ introduces
an additional nonlinear coupling of the transverse and longitudinal dynamics while preserving the
Lyapunov dissipation identity. In particular, note that there are no constraints on the sign of
$\psi$.

\begin{example}[Classical Bloch model]
	The classical Bloch relaxation term is recovered as the special case
	\begin{align*}
		 \beta \equiv 0, \quad
		 \psi \equiv 0, \quad
		a \equiv \dfrac{1}{T_{2}}, \quad
		c \equiv \frac{1}{T_{1}},
	\end{align*}
    which clearly satisfy the assumptions of Proposition \ref{prop:construction_admissible_fields}. 
\end{example}

\begin{example}[Non-exponential longitudinal magnetization dynamics]\label{ex:nonexplong}
Magnetization Transfer principally affects the dynamics of the longitudinal magnetization $M_3$, causing non-exponential behavior \cite{prantner2008magnetization}. 
This leads to the question whether $\dot{M}_3$ can be approximated more accurately with the proposed framework than with the standard 
Bloch model $\dot{M}_3=(M^*-M_3)/T_{1}$. Equation 
\eqref{eq:construction} suggests the generalized form
\begin{align*}
    \dot{M}_3=(M^{\ast}-M_3) \mu_{1}(M_{r}, M_{3}), \quad \mu_{1}(s, x_{3}) := c(s, x_{3}) - s^{2} \psi(s, x_{3}).
\end{align*}
Note that $\mu_1(M_r, M_3)$ acts here as an effective longitudinal relaxation rate. Assuming $\psi=0$ and decoupled relaxation dynamics, 
that is, $\mu_1(M_{r}, M_{3}) = \mu_1(M_3)$, Proposition~\ref{prop:construction_admissible_fields} leaves ample space for an admissible choice of $\mu_1$.

We note that the class of functions $\mu_{1}(x_3) =\kappa_1+\kappa_2(M^*-x_3)^{\kappa_3}$, where $x_{3} \leq M^{\ast}$, $\kappa_{1} > 0$ and $\kappa_{2}, \kappa_{3} 
\geq 0$, is \footnote{Since only the regime $x_3 \leq M^*$ is physically relevant, the admissibility conditions for the case $\kappa_{3} > 0$ can be 
relaxed in a straightforward way, and it can be shown that all trajectories starting below $M^*$ converge to $M^{\ast}$. 
Alternatively, defining $\mu_1(x_3)\equiv\kappa_1$ for $x_3\ge M^*$ allows a direct application of 
Theorem~\ref{thm:admissible_relaxation_fields} to the resulting map $\zeta$.}{admissible}. The resulting initial value problem
\begin{align*}
    \begin{cases}
        \dot M_{3} = (M^{\ast} - M_{3}) \left( \kappa_1+\kappa_2(M^{\ast}-M_{3})^{\kappa_3} \right), \\ 
        M_{3}(0) = M_3^0, \\
    \end{cases}
\end{align*} 
where $M_{3}^{0} \in (-\infty, M^{\ast})$ is an initial condition, is analytically solved for $\kappa_{3} >0$ by 
\begin{align*} 
    M_3(t) = M^* - (M^*-M_3^0)\left(\frac{\kappa_1 e^{-\kappa_1\kappa_3 t}} { \kappa_1+\kappa_2(M^*-M_3^0)^{\kappa_3}(1-e^{-\kappa_1\kappa_3 t}) }\right)^{1/k_3}. 
\end{align*}
For $\kappa_2 = 0$ and $\kappa_3 = 1$ the standard Bloch equation is recovered. More generally, for $\kappa_2,\kappa_3 > 0$ the effective longitudinal relaxation rate $\mu_1$ becomes state-dependent with the special case of linear dependency when $\kappa_3 = 1$. From a pragmatic point of view, the choice 
\begin{align*} 
\dot{M}_r &= -M_r/T_2, \\
\dot{M}_{\theta} &=\omega, \\
\dot{M}_3 & = (M^*-M_3)\left(\kappa_1+\kappa_2(M^*-M_3)^{\kappa_3}\right)
\end{align*}
with initial condition $M_{3}^{0} < M^{\ast}$, where $\omega $ is a constant, seems thus to be particularly attractive to describe complex 
longitudinal magnetization dynamics typical of magnetization transfer effects. 
\end{example}

\begin{example}[Coupled transverse-longitudinal relaxation] Let 
    \begin{align*} 
       \beta \equiv 0, \qquad  \psi \equiv \lambda, \qquad a \equiv \kappa_2,
        \qquad c \equiv \kappa_1, 
    \end{align*} 
    where $\kappa_{1}, \kappa_{2}, \lambda > 0$ are constants. The dynamics in cylindrical coordinates
    become 
    \begin{align*} 
        \dot M_{r} &= -\left( \kappa_2 + \lambda(M_{3} - M^{\ast})^{2} \right)M_{r}, \\ 
        \dot M_{3} &= -\left( \kappa_1 - \lambda M_{r}^{2} \right)(M_{3}-M^{\ast}),\\ 
        \dot M_{\theta} &=0. 
    \end{align*} 
    Unlike the classical Bloch equations, the relaxation rates are no longer independent. The 
    transverse decay rate increases with the squared
    longitudinal deviation from equilibrium, while the longitudinal recovery rate now also depends on the
    squared magnitude of the transverse magnetization. In particular, states that are far from equilibrium 
    in the longitudinal direction exhibit enhanced transverse relaxation. 

    This example illustrates how the coupling term $\psi$ generates nonlinear interactions between
    transverse and longitudinal relaxation channels while preserving all admissibility conditions. Such
    behavior cannot be represented within the classical Bloch model, where the two relaxation mechanisms
    are completely decoupled.
\end{example}

Much more elaborate models can be constructed within this framework. For instance, one may combine 
state-dependent relaxation and transverse-longitudinal coupling by taking $a(s,x_{3}) = \kappa_2 + s^{2}$ 
together with a nontrivial choice of $\psi$. More generally, Proposition~\ref{prop:construction_admissible_fields} 
permits admissible vector fields generated from a large class of nonlinear functions of $s$ and $x_{3} - M^{\ast}$, 
including polynomial, fractional-power, logarithmic, and mixed forms. This flexibility enables the construction of 
highly nonlinear relaxation models while preserving equivariance, the uniqueness and stability of the thermal 
equilibrium, and strict Lyapunov dissipation.

\section{Discussion and Conclusions}
By separating the rotational (Larmor torque) and relaxation components of the classical Bloch equation and imposing admissibility requirements on the latter, we derived the most general form of the relaxation vector field $\bF$ consistent with the proposed physical assumptions. Specifically, based on the 
three postulates of $(i)$ rotational equivariance around the longitudinal axis, $(ii)$ global stability of the thermal equilibrium state, and $(iii)$ 
strict monotonic decrease of the introduced potential energy-like function $V$, 
we obtained a complete characterization of the admissible vector fields $\bF$ (Theorem \ref{thm:admissible_relaxation_fields}).

This characterization reveals several important structural properties. First, when expressed in the form $\bF(x)=A(x)x+b(x)$, see \eqref{eq:admissible_relaxation_field} and even more explicitly \eqref{eq:construction}, the admissible dynamics naturally exhibit a diagonal-plus-skew-symmetric structure. The diagonal component governs dissipative processes, while the skew-symmetric component generates rotations around the longitudinal axis. From a physical perspective, these additional rotational contributions may be interpreted as effective off-resonance effects arising from material-specific interactions. 
A direct consequence of the equivariance requirement is that the components of $\bF$ are independent of the azimuthal (phase) angle. This is a natural outcome, as any explicit phase dependence would violate rotational symmetry around the longitudinal axis. However, the admissible vector fields may still depend on the state of the system through the longitudinal coordinate and the transverse magnitude. Note, in particular, the explicit dependence on the radial and longitudinal state in \eqref{eq:construction}. Such state dependencies extend the descriptive capabilities of the classical Bloch model by allowing coupled relaxation dynamics that vary across different magnetization states, potentially including complex magnetization transfer and/or exchange effects and other molecular interactions.

More generally, Theorem \ref{thm:admissible_relaxation_fields} significantly constrains the class of admissible relaxation models: rather than considering arbitrary nonlinear vector fields, future developments can focus on a restricted family described by the theorem and whose structure is directly linked to the imposed physical principles.
From a practical perspective, we also outlined a constructive procedure for parameterizing a large family of explicit admissible forms of $\bF$, see Proposition \ref{prop:construction_admissible_fields}. The characterization theorem can therefore be viewed not only as a theoretical result but also as a guide for future model construction and effective-medium formulations in magnetic resonance. 

We note that the examples presented here are not intended to be exhaustive and can be extended in several directions. In particular, analogous to the modelling of non-exponential longitudinal relaxation proposed in Example~\ref{ex:nonexplong}, a similar approach could be employed to describe non-exponential transverse relaxation. Such behavior has been reported, for example, in tissue with increased iron deposition, where deviations from simple mono-exponential transverse relaxation are commonly observed \cite{qin2017characterization,jensen2002theory,oliveira2025vivo}. These and other extensions suggest that the proposed framework may provide a flexible basis for capturing a broad range of non-standard MR dynamics. Further investigation of specific parameterizations and their experimental validation is left for future work.

An important limitation of the present study is that the three physical requirements introduced in Section \ref{sec:physical_constraints} are a modeling choice. While we believe that they capture the essential properties of the classical Bloch dynamics while leaving room for richer formulations, alternative sets of assumptions could equally be considered. Depending on the intended application, different requirements may be adopted either to further constrain the admissible model class or to enlarge it. In this sense, the principal contribution of this work is not the particular choice of assumptions, but rather the general strategy of deriving admissible magnetization dynamics from explicitly stated physical principles \cite{golubitsky2012singularities,field2007dynamics}.\\

In conclusion, we introduced a general admissibility framework for magnetization dynamics that recovers the classical Bloch equation as a special case while allowing for a broader class of physically consistent relaxation mechanisms. The resulting formulation remains compact, is grounded in established MR theory, and provides a principled foundation for the development of effective medium models potentially capable of describing complex materials, including biological tissue. By shifting the focus from prescribing a specific relaxation law to characterizing the class of admissible dynamics, the proposed framework offers a new perspective on the modeling of magnetic resonance phenomena in heterogeneous media.

\appendix
\section{Appendix}
\label{sec:appendix}
In this appendix we provide the precise mathematical formulation and proofs of all statements in the
main body of the paper. Many of the results, e.g., the form of vector fields that are equivariant
with respect to rotations around $\be_{3}$, are well-known; though the specific regularity requirements
we are interested in are more difficult to find. In fact, the study of vector fields on smooth
manifolds that are equivariant with respect to the action of a general (compact) Lie group is a well-studied
topic in both mathematics and physics, e.g., see \cite{golubitsky2012singularities,field2007dynamics}. 
Since we only need very basic results, which can be established without any advanced 
results from Lie group theory, we provide elementary proofs of all statements to keep the work self-contained. 

\subsection{Characterization of \texorpdfstring{$\mathrm{SO}(2)$}{}-equivariant vector fields.}
In this section we characterize all continuously differentiable vector fields on $\RR^{3}$ that are
equivariant with respect to rotations around $\bm{e}_{3}$. To be more precise, let $\rhd:
\mathrm{SO}(2) \times \RR^{3} \rightarrow \RR^{3}$ denote the left group action of $\mathrm{SO}(2)$
on $\RR^{3}$ defined by 
\begin{align*}
	R_{2}(\theta) \rhd x := 
	\begin{pmatrix}
		\cos(\theta) & - \sin(\theta) & 0 \\
		\sin(\theta) &  \cos(\theta) & 0 \\
		0 & 0 & 1
	\end{pmatrix}x, \quad
	R_{2}(\theta) := 
	\begin{pmatrix}
		\cos(\theta) & - \sin(\theta) \\
		\sin(\theta) &  \cos(\theta)
	\end{pmatrix},
\end{align*}
where $\theta \in [0, 2\pi)$. We seek a characterization of all continuously differentiable vector
fields $\bF: \RR^{3} \rightarrow \RR^{3}$ such that $\bF(R \rhd x) = R \rhd \bF(x)$ for all $x \in
\RR^{3}$ and $R \in \mathrm{SO}(2)$. 

We start by giving a characterization of all continuously differentiable vector fields on $\RR^{3} 
\setminus \text{span}\{\bm{e}_{3}\}$ that are equivariant with respect to $\rhd$. 

\begin{lemma}[$\rhd$-equivariant vector fields on $\RR^{3} \setminus \text{span}\{\be_{3}\}$]
    \label{lemma:equivariant_vector_field}
    Let $\bF: \RR^{3} \setminus \mathrm{span}\{\be_{3}\} \rightarrow \RR^{3}$ be a continuously
    differentiable vector field. Then $\bF$ is $\rhd$-equivariant if and only if there exists
    $\alpha, \beta, \zeta \in C^{1}\left( (0, \infty) \times \RR \right)$ such that 
	\begin{align}
		\label{eq:appendix_equivariant_vector_field}
		\bF(x) = 
		\begin{pmatrix}
			\alpha(r,x_{3}) & -\beta(r, x_{3}) & 0 \\
			\beta(r, x_{3}) & \alpha(r, x_{3}) & 0 \\
			0 & 0 & 0 
		\end{pmatrix}x + 
		\begin{pmatrix}
			0 \\ 0 \\ \zeta(r, x_{3})
		\end{pmatrix},
	\end{align}
	where we have abbreviated $r = r(x_{1}, x_{2}) = \sqrt{ x_{1}^{2} + x_{2}^{2} }$. 
	\begin{proof}
        A direct computation shows that every vector field of the form in
        \eqref{eq:appendix_equivariant_vector_field} defines a continuously differentiable vector
        field on $\RR^{3} \setminus \text{span}\{\be_{3}\}$ and is equivariant with respect to
        $\rhd$. Conversely, suppose $\bF: \RR^{3} \setminus \text{span}\{\be_{3}\} \rightarrow
        \RR^{3}$ is a continuously differentiable vector field equivariant w.r.t. $\rhd$. Let $x \in
        \RR^{3} \setminus \text{span}\{\be_{3}\}$ be arbitrary and denote its cylindrical
        coordinates by $(r, \theta, x_{3}) \in (0, \infty) \times [0, 2 \pi) \times \RR$, i.e.,
        $x_{1} = r \cos \theta$ and $x_{2} = r \sin \theta$. Then 
		\begin{align*}
			\bF(x) = \bF \left( R_{2}(\theta) \rhd (r, 0, x_{3}) \right) = 
			R_{2}(\theta) \rhd \bF(r, 0, x_{3}),
		\end{align*}
		where in the last line we used the equivariance of $\bF$. 
		
        Next, note that the components of $(s,x_{3}) \mapsto \bF(s, 0, x_{3})$ are continuously
        differentiable on $(0, \infty) \times \RR$, since $\bF$ is continuously differentiable on
        $\RR^{3} \setminus \text{span}\{\be_{3}\}$. Therefore, writing 
		\begin{align*}
			\bF(s, 0, x_{3}) = 
			\begin{pmatrix} 
				f_{1}(s, x_{3}) \\ f_{2}(s, x_{3}) \\  f_{3}(s, x_{3}) 
			\end{pmatrix},
		\end{align*}
		where $f_{1}, f_{2}, f_{3} \in C^{1}((0, \infty) \times \RR)$, we see that 
		\begin{align*}
			\bF(x) = R_{2}(\theta) \rhd \bF(r, 0, x_{3}) = 
			\begin{pmatrix}
				\displaystyle \frac{f_{1}(r,x_{3})}{r} x_{1} - \frac{f_{2}(r,x_{3})}{r} x_{2} \\[2ex]
				\displaystyle \frac{f_{1}(r,x_{3})}{r} x_{2} + \frac{f_{2}(r,x_{3})}{r} x_{1} \\[2ex]
				f_{3}(r,x_{3})
			\end{pmatrix}.
		\end{align*}
		Hence $\bF$ is of the form in \eqref{eq:appendix_equivariant_vector_field} with 
		\begin{align*}
			\alpha(s, x_{3}) =  \frac{f_{1}(s,x_{3})}{s}, \qquad
			\beta(s, x_{3}) = \frac{f_{2}(s,x_{3})}{s}, \qquad
			\zeta(s, x_{3}) = f_{3}(s,x_{3}).
		\end{align*}
        Since these mappings are continuously differentiable on $(0, \infty) \times \RR$, this
        proves the claim. 
	\end{proof}
\end{lemma}
\begin{remark}[Notation]
	As in the main body of the paper,  we shall frequently denote the vector field in 
	\eqref{eq:appendix_equivariant_vector_field} by $\bF_{(\alpha, \beta, \zeta)}$. 
\end{remark}
To extend the result of Lemma \ref{lemma:equivariant_vector_field} to vector fields that are
continuously differentiable on all of $\RR^{3}$, we need additional conditions on $\alpha, \beta$
and $\zeta$. The next two lemmas provide these conditions.
\begin{lemma}[I. Differentiability on $\text{span}\{ \be_{3} \}$]
	\label{lemma:differentiability_e3}
    Let $g \in C^{1}([0, \infty) \times \RR)$ be arbitrary and define $G: \RR^{3} \rightarrow \RR$
    by $G(x) := g(r, x_{3})$. Then $G \in C^{1}(\RR^{3})$ if and only if $\partial_{s}g(0, x_{3}) =
    0$ for all $x_{3} \in \RR$.  
	\begin{proof}
        First observe that we always have $G \in C^{1} \left( \RR^{3} \setminus
        \text{span}\{\be_{3}\} \right)$, since $(x_{1}, x_{2}) \mapsto r(x_{1}, x_{2})$ is
        continuously differentiable on $\RR^{3} \setminus \text{span}\{\bm{e}_{3}\}$, and $g \in
        C^{1}([0, \infty) \times \RR)$. Furthermore, note that $\partial_{x_{3}}G$ exists on all of
        $\RR^{3}$, and is continuous, since (again) $g \in C^{1}([0, \infty) \times \RR)$.
        Therefore, what remains is the analysis of $\partial_{x_{1}}G(0, 0, x_{3})$ and
        $\partial_{x_{2}}G(0, 0, x_{3})$.
		
        Let us first consider the existence of $\partial_{x_{1}}G(0, 0, x_{3})$. To this end, let
        $h, x_{3} \in \RR$ be arbitrary, and note that 
		\begin{align*}
            \frac{ G(h, 0, x_{3}) - G(0, 0, x_{3}) }{h} = \frac{ g(\vert h \vert, x_{3}) - g(0,
            x_{3}) }{ h}. 
		\end{align*}
		Therefore, for $h >0$, we have
		\begin{align}
			\label{eq:lower_limit}
			\lim_{h \downarrow 0}\frac{ G(h, 0, x_{3}) - G(0, 0, x_{3}) }{h} &= 
			\lim_{h \downarrow 0}\frac{ g(h, x_{3}) - g(0, x_{3}) }{ h} = 
			\partial_{s} g(0, x_{3}),
		\end{align}
		since $g$ is differentiable at $(0, x_{3})$. Similarly, 
		\begin{align}
			\label{eq:upper_limit}
			\lim_{h \uparrow 0}\frac{ G(h, 0, x_{3}) - G(0, 0, x_{3}) }{h} &= 
			\lim_{h \downarrow 0} - \frac{ g(h, x_{3}) - g(0, x_{3}) }{h} = 
			- \partial_{s} g(0, x_{3}).
		\end{align}
        In other words, the upper and lower limits always exist. Therefore, $\partial_{x_{1}}G(0, 0,
        x_{3})$ exists if and only if \eqref{eq:lower_limit} and \eqref{eq:upper_limit} coincide,
        which is equivalent to $\partial_{s} g(0, x_{3}) = 0$. The same argument applies to the
        existence of $\partial_{x_{2}}G(0, 0, x_{3})$. Consequently, the partial derivatives
        $\partial_{x_{1}}G$ and $\partial_{x_{2}}G$ exist at every point of $\RR^{3}$ iff
        $\partial_{s} g(0, \cdot) \equiv 0$. In particular, if the latter holds, then 
		$
			\partial_{x_{1}}G(0, 0, x_{3}) = \partial_{x_{2}}G(0, 0, x_{3})= 0. 
		$
		
        It remains to be shown that if $\partial_{s}g(0, \cdot) \equiv 0$, then the partial
        derivatives $\partial_{x_{1}}G$ and $\partial_{x_{2}}G$ are continuous on
        $\text{span}\{\be_{3}\}$. To this end, assume $\partial_{s}g(0, \cdot) \equiv 0$, and let $j
        \in \{1, 2\}$ and $x \in \RR^{3} \setminus \text{span}\{\be_{3}\}$ be arbitrary, then 
		\begin{align*}
			\partial_{x_{j}} G(x) = \partial_{s} g(r, x_{3}) \frac{ x_{j} }{r}.
		\end{align*}
		Consequently, 
		$
			\left \vert \partial_{x_{j}} G(x) \right \vert \leq 
			\left \vert \partial_{s} g(r, x_{3}) \right \vert.
		$
		Hence 
		\begin{align*}
            \lim_{x \mapsto (0, 0, \tilde x_{3})}\partial_{x_{j}} G(x) = 0, \quad \forall \tilde
            x_{3} \in \RR,
		\end{align*}
        since $\partial_{s}g$ is (in particular) continuous at $(0, \tilde x_{3})$. Therefore, all
        partial derivatives of $G$ exist and are continuous on $\RR^{3}$, which implies that $G \in
        C^{1}(\RR^{3})$.  
	\end{proof}
\end{lemma}

\begin{lemma}[II. Differentiability on $\text{span}\{ \be_{3} \}$]
	\label{lemma:differentiability_e3_v2}
    Let $g \in C^{1}((0, \infty) \times \RR) \cap C([0, \infty) \times \RR)$ and $i \in \{1, 2\}$.
    Define $G: \RR^{3} \rightarrow \RR$ by $G(x) := x_{i}g(r, x_{3})$. If 
	\begin{align}
		\label{eq:weak_radial_diff_conditions}
		\lim_{(s, x_{3}) \rightarrow (0, \tilde x_{3})} s \partial_{s}g(s, x_{3}) = 0, \quad
		\lim_{(s, x_{3}) \rightarrow (0, \tilde x_{3})} s \partial_{x_{3}}g(s, x_{3}) = 0, 
	\end{align}
	for all $\tilde x_{3} \in \RR$, then $G \in C^{1}(\RR^{3})$. 
	\begin{proof}
        It is clear that $G \in C^{1}\left (\RR^{3} \setminus \text{span} \{ \be_{3} \} \right)$, since
        $\left( x \mapsto (r(x_{1}, x_{2}), x_{3}) \right) \in C^{1}\left (\RR^{3} \setminus
        \text{span} \{ \be_{3} \} \right)$. Next, we show that all partial derivatives on 
        $\text{span}\{ \be_{3} \}$ exist and are continuous. To this end, let $j \in \{1, 2\}$ 
        be arbitrary, then 
		\begin{align*}
			\lim_{h \rightarrow 0} \frac{G( h \be_{j} + x_{3} \be_{3}) - G(0, 0, x_{3})}{h} =
			\lim_{h \rightarrow 0} \delta_{ij} g(\vert h \vert, x_{3}) = \delta_{ij} g(0, x_{3}),
		\end{align*}
        since $g$ is continuous at $(0, x_{3})$. Hence $\partial_{x_{j}}G$ exists on $\RR^{3}$ for
        $j \in \{1, 2 \}$. Furthermore, we claim that these partial derivatives are continuous on
        $\text{span}\{\be_{3}\}$. To see this, first observe that 
		\begin{align*}
            \partial_{x_{j}}G(x) = \delta_{ij} g(r, x_{3}) + \frac{ x_{i} x_{j} }{r}
            \partial_{s}g(r, x_{3}), \quad x \in \RR^{3} \setminus \text{span} \{ \be_{3} \}.  
		\end{align*}
		Now, let $\tilde x_{3} \in \RR$ be arbitrary, then it follows that 
		\begin{align*}
            \left \vert \partial_{x_{j}}G(0, 0, \tilde x_{3}) - \partial_{x_{j}}G(x) \right \vert
            \leq \delta_{ij} \left \vert g(0, \tilde x_{3}) - g(r, x_{3}) \right \vert 
            + \left \vert r \partial_{s}g(r, x_{3}) \right \vert \rightarrow 0, 
		\end{align*}
        as $x \rightarrow (0, 0, \tilde x_{3})$, since $g$ is continuous at $(0, \tilde x_{3})$ and
        by the first condition in \eqref{eq:weak_radial_diff_conditions}. Hence $\partial_{x_{j}}G
        \in C(\RR^{3})$ for $j \in \{1, 2 \}$. 
		
        It remains to show that $\partial_{x_{3}}G$ exists on $\text{span}\{\be_{3}\}$ and is
        continuous. Since $G$ is zero on $\text{span}\{\be_{3}\}$, $\partial_{x_{3}}G(0, 0, x_{3})$
        exists and is zero for all $x_{3} \in \RR$. Furthermore, 
		\begin{align*}
			\left \vert \partial_{x_{3}}G(x) - \partial_{x_{3}}G(0, 0, \tilde x_{3}) \right \vert  = 
			\left \vert x_{i} \partial_{x_{3}} g(r, x_{3}) \right \vert \leq
			r \left \vert \partial_{x_{3}} g(r, x_{3}) \right \vert \rightarrow 0
		\end{align*}
        as $x \rightarrow (0, 0, \tilde x_{3})$ by the second condition in
        \eqref{eq:weak_radial_diff_conditions}. Therefore, $\partial_{x_{3}}G \in C(\RR^{3})$.
        Altogether, we now conclude that $G \in C^{1}(\RR^{3})$. 
	\end{proof}
\end{lemma}

Lemma's \ref{lemma:equivariant_vector_field}, \ref{lemma:differentiability_e3} and
\ref{lemma:differentiability_e3_v2} together can be used to provide a complete characterization of
continuously differentiable $\rhd$-equivariant vector fields on $\RR^{3}$. 
\begin{proposition}[$\rhd$-equivariant vector fields on $\RR^{3}$]
    \label{prop:characterization_equivariant_vector_fields}
    Let $\bF: \RR^{3} \rightarrow \RR^{3}$ be a continuously differentiable vector field. Then $\bF$
    is $\rhd$-equivariant if and only if there exists $\alpha, \beta \in C^{1}\left( (0, \infty)
    \times \RR \right) \cap C([0, \infty) \times \RR)$, $\zeta \in C^{1}([0, \infty) \times \RR)$
    such that 
	\begin{align*}
		\bF(x) = 
		\begin{pmatrix}
			\alpha(r,x_{3}) & -\beta(r, x_{3}) & 0 \\
			\beta(r, x_{3}) & \alpha(r, x_{3}) & 0 \\
			0 & 0 & 0 
		\end{pmatrix}x + 
		\begin{pmatrix}
			0 \\ 0 \\ \zeta(r, x_{3})
		\end{pmatrix},
	\end{align*}
	and
	\begin{enumerate}
		\item[(\textbf{D1})] 
            $s \partial_{s} \alpha(s, x_{3}) \rightarrow 0$, $s
            \partial_{x_{3}} \alpha(s, x_{3}) \rightarrow 0$ as $(s, x_{3}) \rightarrow (0, \tilde
            x_{3})$, $\forall \tilde x_{3} \in \RR$, 
		\item[(\textbf{D2})] 
            $s \partial_{s} \beta(s, x_{3}) \rightarrow 0$, $s
            \partial_{x_{3}} \beta(s, x_{3}) \rightarrow 0$ as $(s, x_{3}) \rightarrow (0, \tilde
            x_{3})$, $\forall \tilde x_{3} \in \RR$, 
		\item[(\textbf{D3})] $\partial_{s} \zeta(0, x_{3}) =0$, $\forall x_{3} \in \RR$.
	\end{enumerate}
	\begin{proof}
        ``$\Rightarrow$'' Assume $\bF \in C^{1}(\RR^{3}; \RR^{3})$ is $\rhd$-equivariant, then there
        exist $\alpha, \beta, \zeta \in C^{1}((0, \infty) \times \RR)$ such that $\bF =
        \bF_{(\alpha, \beta, \zeta)}$ on $\RR^{3} \setminus \text{span} \{ \be_{3} \}$ by Lemma
        \ref{lemma:equivariant_vector_field}. We start by showing that $\alpha$ admits a continuous
        extension to $[0, \infty) \times \RR$. To this end, note that 
		\begin{align*}
            \alpha(r, x_{3}) = \frac{x_{1} F_{1}(x) + x_{2} F_{2}(x)}{ r^{2}}, \quad x \in \RR^{3}
            \setminus  \text{span}\{ \be_{3} \}. 
		\end{align*}
		In particular,
		\begin{align}
			\label{eq:alpha_F1}
            \alpha(s, x_{3}) = \frac{ F_{1}(s, 0, x_{3}) }{s}, \quad (s, x_{3}) \in (0, \infty)
            \times \RR. 
		\end{align}
        Extend $\alpha$ by setting $\alpha(0, x_{3}) = \partial_{x_{1}} F_{1}(0, 0, x_{3})$. We
        claim that this makes $\alpha$ continuous on $[0, \infty) \times \RR$. To see this, let
        $\tilde x_{3} \in \RR$ be arbitrary, then 
		\begin{align*}
			& \left \vert \alpha(0, \tilde x_{3}) - \alpha(s, x_{3}) \right \vert \leq 
			\left \vert \partial_{x_{1}}F_{1}(0, 0, \tilde x_{3}) - \partial_{x_{1}}F_{1}(0, 0, x_{3}) \right \vert + 
			\left \vert \partial_{x_{1}}F_{1}(0, 0, x_{3}) - \frac{ F_{1}(s, 0, x_{3})}{s} \right \vert.		
		\end{align*}
        The first term goes to zero as $(s, x_{3}) \rightarrow (0, \tilde x_{3})$, since
        $\partial_{x_{1}}F_{1}  \in C(\RR^{3})$. To see that the second term goes to zero,
        observe that
		  \begin{align}
		 	\label{eq:dx1F1_integral}
		 	\left \vert \partial_{x_{1}}F_{1}(0, 0, x_{3}) - \frac{ F_{1}(s, 0, x_{3})}{s} \right \vert 
			\leq
            \int_{0}^{1} \left \vert \partial_{x_{1}}F_{1}(0, 0, x_{3}) - \partial_{x_{1}}F_{1}(ts,
            0, x_{3}) \right \vert \dd t. 
		  \end{align}
        Here we used that $F_{1}(0, 0, \cdot) \equiv 0$, which follows since any $\rhd$-equivariant
        vector field leaves $\text{span}\{ \be_{3} \}$ invariant. The righthand-side of
        \eqref{eq:dx1F1_integral} goes to zero as $(s, x_{3}) \rightarrow (0, \tilde x_{3})$, since
        $(s, x_{3}) \mapsto \partial_{x_{1}}F_{1}(s, 0, x_{3})$ is (uniformly) continuous in a
        closed ball around $(0, \tilde x_{3})$. Therefore, since $\tilde x_{3} \in \RR$ was
        arbitrary, we conclude that $\alpha \in C([0, \infty) \times \RR)$. 
		 
         The properties in (\textbf{D1}) are easily established: it follows directly from
         \eqref{eq:alpha_F1} that 
	    \begin{align*}
            \lim_{(s, x_{3}) \rightarrow (0, \tilde x_{3}) } s \partial_{x_{3}} \alpha(s, x_{3}) &=
            \partial_{x_{3}} F_{1}(0, 0, \tilde x_{3}) = 0, 
	    \end{align*}
        since $F_{1} \in C^{1}(\RR^{3})$ and $F_{1}(0, 0, \cdot) \equiv 0$. Similarly,
        differentiation of \eqref{eq:alpha_F1} shows that 
		\begin{align*}
			s \partial_{s} \alpha(s, x_{3}) = 
            \partial_{x_{1}} F_{1}(s, 0, x_{3}) - \alpha(s, x_{3}), \quad (s, x_{3}) \in (0, \infty)
            \times \RR. 
		\end{align*}
        Therefore, since the terms on the righthand-side are continuous at $(0, 0, \tilde x_{3})$
        and $(0, \tilde x_{3})$, respectively, and $\alpha(0, \tilde x_{3}) = \partial_{x_{1}}
        F_{1}(0, 0, \tilde x_{3})$, we conclude that 
		\begin{align*}
            \lim_{(s, x_{3}) \rightarrow (0, \tilde x_{3}) } s \partial_{s} \alpha(s, x_{3}) = 0,
            \quad \forall \tilde x_{3} \in \RR. 
		\end{align*}
		Now, since 
		\begin{align*}
            \beta(r, x_{3}) = \frac{ x_{1} F_{2}(x) - x_{2} F_{1}(x) }{ r^{2}}, \quad x \in \RR^{3}
            \setminus  \text{span}\{ \be_{3} \},
		\end{align*} 
        the same line of reasoning shows that $\beta$ can be continuously extended to $[0,
        \infty) \times \RR$ and satisfies (\textbf{D2}).	
		
        It remains to show that $\zeta$ admits a continuously differentiable extension to $[0,
        \infty) \times \RR$  and (\textbf{D3}) is satisfied. We extend $\zeta$ by setting $\zeta(0,
        x_{3}) = F_{3}(0, 0, x_{3})$. Then $F_{3}(x) = \zeta(r, x_{3})$ at all points of $\RR^{3}$. 
        Furthermore, with this choice, all partial derivatives of $\zeta$ exist at points $(0,
        x_{3})$ for $x_{3} \in \RR$. Indeed, since $F_{3}$ is (continuously) differentiable, we have
	    \begin{align*}
            \partial_{s} \zeta(0, x_{3}) = \lim_{h \downarrow 0} \frac{ F_{3}(h, 0, x_{3}) -
            	F_{3}(0, 0, x_{3}) }{h} = \partial_{x_{1}} F_{3}(0, 0, x_{3}), \quad
		    \partial_{x_{3}}\zeta(0, x_{3}) = \partial_{x_{3}}F_{3}(0, 0, x_{3}).
    	\end{align*} 
        Therefore, since $\zeta(s, x_{3}) = F_{3}(s, 0, x_{3})$ for $(s, x_{3}) \in (0, \infty)
        \times \RR$, we conclude that
		\begin{align*}
			\partial_{s} \zeta(s, x_{3}) = \partial_{x_{1}} F_{3}(s, 0, x_{3}),  \quad
			\partial_{x_{3}} \zeta(s, x_{3}) = \partial_{x_{3}} F_{3}(s, 0, x_{3}), \
			\forall (s, x_{3}) \in [0, \infty) \times \RR. 
		\end{align*}
        Consequently, $\partial_{s} \zeta, \partial_{x_{3}} \zeta \in C([0, \infty) \times \RR)$,
        since $F_{3} \in C^{1}(\RR^{3})$, and therefore $\zeta \in C^{1}([0, \infty) \times \RR)$.
        We may now apply Lemma \ref{lemma:differentiability_e3} to conclude that $\partial_{s}
        \zeta(0, \cdot) =0$.
				
        ``$\Leftarrow$'' Conversely, assume  there exists $\alpha, \beta \in C^{1}\left( (0, \infty)
        \times \RR \right) \cap C([0, \infty) \times \RR)$, $\zeta \in C^{1}([0, \infty) \times
        \RR)$ such that $\bF$ is of the form in \eqref{eq:appendix_equivariant_vector_field} and
        properties (\textbf{D1})--(\textbf{D3}) are satisfied. Properties (\textbf{D1}),
        (\textbf{D2}) and Lemma \ref{lemma:differentiability_e3_v2} guarantee that the mappings $x
        \mapsto x_{j} \alpha(r, x_{3})$ and $x \mapsto x_{j} \beta(r, x_{3})$ are continuously
        differentiable on $\RR^{3}$ for $j \in \{1, 2\}$. Similarly, property (\textbf{D3}) and
        Lemma \ref{lemma:differentiability_e3} guarantee that $\left( x \mapsto \zeta(r, x_{3})
        \right) \in C^{1}(\RR^{3})$. Therefore, $\bF \in C^{1}(\RR^{3}; \RR^{3})$. Moreover, $\bF =
        \bF_{(\alpha, \beta, \zeta)}$ is $\rhd$-equivariant by Lemma
        \ref{lemma:equivariant_vector_field} (as equivariance on $\text{span}\{ \be_{3} \}$ is
        clear). 
	\end{proof}
\end{proposition}

\subsection{Asymptotic stability and non-degeneracy of \texorpdfstring{$\bM^{\ast}$}{}}
In this section we characterize when a $\rhd$-equivariant vector field has $\bM^{\ast} = (0, 0,
M^{\ast})$ as its unique, non-degenerate (hyperbolic), locally stable equilibrium. Global stability
will be addressed in the next section. 

\begin{lemma}[Uniqueness equilibrium]
    \label{lemma:unique_equilibrium}
    Let $\bF = \bF_{(\alpha, \beta, \zeta)}$ be a continuously differentiable $\rhd$-equivariant
    vector field defined by parameters $\alpha, \beta \in C^{1}\left( (0, \infty) \times \RR \right)
    \cap C([0, \infty) \times \RR)$ and $\zeta \in C^{1}([0, \infty) \times \RR)$ satisfying
    (\textbf{D1})--(\textbf{D3}). Then $\bM^{\ast}$ is the unique zero of $\bF$ if and only if the
    following conditions are satisfied: 
    \begin{enumerate}
        \item[$(i)$]  $\zeta(0,M^{\ast}) = 0$, $\zeta(0, x_{3}) \not = 0, \ \forall x_{3} \in \RR
            \setminus \{M^{\ast} \}$.
        \item[$(ii)$] $\alpha(s,x_{3})^{2} + \beta(s, x_{3})^{2} + \zeta(s, x_{3})^{2} > 0$ for all
            $(s, x_{3}) \in (0, \infty) \times \RR$. 
    \end{enumerate}
    \begin{proof}
        The zeros of $\bF$ are characterized by the follow system of equations: 
        \begin{align*}
            \begin{cases}
                \begin{pmatrix}
                    \alpha(r, x_{3}) & - \beta(r, x_{3}) \\
                    \beta(r, x_{3}) & \alpha(r, x_{3})
                \end{pmatrix} 
                \begin{pmatrix}
                    x_{1} \\ x_{2} 
                \end{pmatrix} = 0, \\[3ex] 
                \zeta(r, x_{3}) = 0. 
            \end{cases}
        \end{align*}
        In particular, for any $x \in \RR^{3}$, the null space of
        \begin{align}
            \label{eq:rotation_like_mat}
                \begin{pmatrix}
                    \alpha(r, x_{3}) & - \beta(r, x_{3}) \\
                    \beta(r, x_{3}) & \alpha(r, x_{3})
                \end{pmatrix} 	
        \end{align}
        is non-trivial iff $\alpha(r,x_{3})^{2} + \beta(r, x_{3})^{2} = 0$, i.e, iff $\alpha(r,
        x_{3}) = \beta(r, x_{3}) = 0$. Therefore, $\bF$ has a zero at $x \in \RR^{3} \setminus
        \text{span} \{ \be_{3} \}$ iff $\alpha(r, x_{3}) = \beta(r, x_{3}) =  \zeta(r, x_{3}) = 0$.
        Equivalently, $\bF$ has no zeros in $\RR^{3} \setminus \text{span} \{ \be_{3} \}$ iff $(ii)$
        is satisfied. Furthermore, since the first two components of $\bF$ are zero on $\text{span}
        \{ \be_{3} \}$, we see that $(i)$ is equivalent to $\bM^{\ast}$ being the only zero on
        $\text{span} \{ \be_{3} \}$. Altogether, this establishes the equivalence. 
    \end{proof}
\end{lemma}

Next, we discuss the local stability of $\bM^{\ast}$. 
\begin{proposition}[Local asymptotic stability]
    \label{prop:local_stable_equilibrium}
    Let $\bF = \bF_{(\alpha, \beta, \zeta)}$ be a continuously differentiable $\rhd$-equivariant
    vector field defined by parameters $\alpha, \beta \in C^{1}\left( (0, \infty) \times \RR \right)
    \cap C([0, \infty) \times \RR)$ and $\zeta \in C^{1}([0, \infty) \times \RR)$ satisfying
    (\textbf{D1})--(\textbf{D3}). Then $\bM^{\ast}$ is the unique equilibrium of $\bF$, and
    moreover, hyperbolic and asymptotically stable, if and only if
    \begin{enumerate}
        \item[$(i)$]  $\zeta(0,M^{\ast}) = 0$, $\zeta(0, x_{3}) \not = 0, \ \forall x_{3} \in \RR
            \setminus \{M^{\ast} \}$.
        \item[$(ii)$] $\alpha(s,x_{3})^{2} + \beta(s, x_{3})^{2} + \zeta(s, x_{3})^{2} > 0$ for all
            $(s, x_{3}) \in (0, \infty) \times \RR$. 
        \item[$(iii)$] $\alpha(0,M^{\ast}) < 0$ and $\partial_{x_{3}}\zeta(0,M^{\ast}) < 0$.
    \end{enumerate}
    \begin{proof}
        The first two conditions $(i)$ and $(ii)$ are equivalent to $\bM^{\ast}$ being the unique
        equilibrium of $\bF$ by Lemma \ref{lemma:unique_equilibrium}. The condition that
        $\bM^{\ast}$ is hyperbolic and asymptotically stable is equivalent to the requirement that
        $\text{Re} \left( \sigma \left( D\bF( \bM^{\ast}) \right) \right) \subset (-\infty, 0)$ by
        the Linearization Theorem for hyperbolic equilibria. To evaluate $D\bF(\bM^{\ast})$, 
        however, we need to be a bit careful, as the parameters $\alpha$ and $\beta$ themselves need
        not be differentiable at $\bM^{\ast}$. 
        
        Since $\bF$ is $\rhd$-equivariant, and $\bM^{\ast}$ is left fixed by $\rhd$, we must have
        \begin{align*}
            R_{3}(\theta) D\bF(\bM^{\ast}) = D\bF(\bM^{\ast})R_{3}(\theta), \quad \forall \theta \in
            [0, 2\pi). 
        \end{align*}
        Evaluation at $\theta = \pi/2$ reveals that $D\bF(\bM^{\ast})$ must be of the following
        form: 
        \begin{align*}
            D\bF(\bM^{\ast}) =             
            \begin{pmatrix}
                p & -q & 0 \\
                q & p & 0 \\
                0 & 0 & w
            \end{pmatrix}, \quad
            p = \partial_{x_{1}}F_{1}(\bM^{\ast}), \
            q = \partial_{x_{1}}F_{2}(\bM^{\ast}), \
            w = \partial_{x_{3}}F_{3}(\bM^{\ast}).
        \end{align*}
        In particular, the characteristic polynomial of $D\bF(\bM^{\ast})$ is given by 
        \begin{align*}
            p(\lambda) = (w - \lambda) \left( (p - \lambda)^{2} + q^{2} \right),
        \end{align*}
        and hence its eigenvalues by 
        \begin{align*}	
            \lambda_{1} = w, \quad
            \lambda_{2} = p - iq, \quad 
            \lambda_{3} = p + iq. 
        \end{align*}
        Consequently, $\text{Re} \left( \sigma \left( D\bF( \bM^{\ast}) \right) \right) \subset
        (-\infty, 0)$ if and only if $w, p < 0$. 
        
        We evaluate $p$ and $q$ using the definition of a partial derivative:
        \begin{align*}
            p = \partial_{x_{1}}F_{1}(\bM^{\ast}) &= 
            \lim_{h \rightarrow 0} \frac{ F_{1}(h, 0, M^{\ast})}{h} = 
            \lim_{h \rightarrow 0} \alpha( \vert h \vert, M^{\ast}) = 
            \alpha(0, M^{\ast}), \\[1ex]
            q = \partial_{x_{1}}F_{2}(\bM^{\ast}) &= 
            \lim_{h \rightarrow 0} \frac{ F_{2}(h, 0, M^{\ast})}{h} = 
            \lim_{h \rightarrow 0} \beta(\vert h \vert, M^{\ast}) = \beta(0, M^{\ast}).
        \end{align*}
        Here we used that $\alpha$ and $\beta$ are continuous at $(0, M^{\ast})$. Furthermore, we
        may directly differentiate $F_{3}$ w.r.t. $x_{3}$, since $\zeta \in C^{1}([0, \infty) \times
        \RR)$, which shows that 
        \begin{align*}
            w = \partial_{x_{3}}F_{3}(\bM^{\ast}) = \partial_{x_{3}}\zeta(0,M^{\ast}).
        \end{align*}
        Therefore, $\text{Re} \left( \sigma \left( D\bF( \bM^{\ast}) \right) \right) \subset
        (-\infty, 0)$ is equivalent to condition $(iii)$. 
    \end{proof}
\end{proposition}
\begin{remark}[Real stable spectrum]
	\label{remark:real_stable_spectrum}
    The proof also shows that $\sigma \left( D\bF( \bM^{\ast}) \right) \subset (-\infty, 0)$ is
    equivalent to $\alpha(0,M^{\ast}) < 0$, $\beta(0, M^{\ast})=0$ and
    $\partial_{x_{3}}\zeta(0,M^{\ast}) < 0$.
\end{remark}

\subsection{Global asymptotic stability of \texorpdfstring{$\bM^{\ast}$}{}}
\label{appendix:global_asymptotic_stability}
In this section we provide sufficient and necessary conditions on $(\alpha, \beta, \zeta)$ such that
the map $V: \RR^{3} \rightarrow [0, \infty)$ defined by $V(x): = \frac{1}{2} \Vert x - \bM^{\ast}
\Vert_{2}^{2}$ is a strict global Lyapunov function for the dynamical system generated by
$\bF_{(\alpha, \beta, \zeta)}$ and its equilibrium $\bM^{\ast}$. Carefully notice here that the
definition of a \emph{global} Lyapunov function implicitly requires $\bM^{\ast}$ to be the unique
equilibrium of $\bF$. 

\begin{proposition}[Strict Lyapunov function]
    \label{prop:strict_lyapunov}
    Let $\bF = \bF_{(\alpha, \beta, \zeta)}$ be a continuously differentiable $\rhd$-equivariant
    vector field defined by parameters $\alpha, \beta \in C^{1}\left( (0, \infty) \times \RR \right)
    \cap C([0, \infty) \times \RR)$ and $\zeta \in C^{1}([0, \infty) \times \RR)$ satisfying
    (\textbf{D1})--(\textbf{D3}). Then $V$ is a strict global Lyapunov function for $\bM^{\ast}$ if
    and only if 
    \begin{itemize}
        \item[$(i)$]  $\zeta(0,M^{\ast}) = 0$, $\zeta(0, x_{3}) \not = 0, \ \forall x_{3} \in \RR
            \setminus \{M^{\ast} \}$.
        \item[$(ii)$] $s^{2} \alpha(s,x_{3}) + (x_{3}-M^{\ast}) \zeta(s, x_{3}) < 0$, for all  $(s,
            x_{3}) \in \left( [0, \infty) \times \RR \right) \setminus \{(0, M^{\ast})\}$. 
    \end{itemize}
    \begin{proof}
        First note that $V( \bM^{\ast}) = 0$, $V(x) > 0$ for $x \in \RR^{3} \setminus \{ \bM^{\ast} \}$, 
        and $V$ is radially unbounded. 
        
        ``$\Rightarrow$'' Assume $V$ is a strict global Lyapunov function for $\bM^{\ast}$. Then
        $\bM^{\ast}$ is the unique equilibrium of $\bF$ (by definition), which in particular implies
        $(i)$ by Lemma \ref{lemma:unique_equilibrium}. Furthermore, $V$ is strictly decreasing along
        non-equilibrium integral curves of $\bF$, which is equivalent to 
        \begin{align}
        	    \label{eq:dVdt}
            \langle \nabla V(x), \bF(x) \rangle = r^{2} \alpha(r, x_{3}) + (x_{3} - M^{\ast})
            \zeta(r, x_{3}) < 0
        \end{align} 
        for all $x \in \RR^{3} \setminus \{ \bM^{\ast} \}$. In turn, this is equivalent to $(ii)$. 
        
        ``$\Leftarrow$'' Conversely, assume conditions $(i)$ and $(ii)$ are satisfied. 
    	Note that $(ii)$ implies 
        \begin{align}
            \label{eq:alpha_beta_zeta_nonzero}
            \alpha(s,x_{3})^{2} + \beta(s, x_{3})^{2} + \zeta(s, x_{3})^{2} > 0, \quad \forall(s,
            x_{3}) \in (0, \infty) \times \RR.
        \end{align} 
        Therefore, conditions $(i)$ and $(ii)$ together imply that $\bM^{\ast}$ is the unique
        equilibrium of $\bF$ by Lemma \ref{lemma:unique_equilibrium}. Since we already established
        that $(ii)$ is equivalent to $\langle \nabla V(x), \bF(x) \rangle < 0$ for all $x \in
        \RR^{3} \setminus \{ \bM^{\ast} \}$ in \eqref{eq:dVdt}, we conclude that $V$ is a strict
        global Lyapunov function for $\bM^{\ast}$. 
    \end{proof}
\end{proposition}

\begin{corollary}[Global asymptotic stability]
    Let $\bF = \bF_{(\alpha, \beta, \zeta)}$ be a continuously differentiable $\rhd$-equivariant
    vector field defined by parameters $\alpha, \beta \in C^{1}\left( (0, \infty) \times \RR \right)
    \cap C([0, \infty) \times \RR)$ and $\zeta \in C^{1}([0, \infty) \times \RR)$ satisfying
    (\textbf{D1})--(\textbf{D3}). Assume the following conditions are satisfied:
    \begin{itemize}
        \item[$(i)$]  $\zeta(0,M^{\ast}) = 0$, $\zeta(0, x_{3}) \not = 0, \ \forall x_{3} \in \RR
            \setminus \{M^{\ast} \}$.
        \item[$(ii)$] $s^{2} \alpha(s,x_{3}) + (x_{3}-M^{\ast}) \zeta(s, x_{3}) < 0$, for all  $(s,
            x_{3}) \in \left( [0, \infty) \times \RR \right) \setminus \{(0, M^{\ast})\}$. 
    \end{itemize}    
    Then $\bF$ is complete, $\bM^{\ast}$ is the unique equilibrium, and moreover, is globally
    asymptotically stable. 
    \begin{proof}
        Conditions $(i)$ and $(ii)$ are equivalent to $V$ being a strict Lyapunov function for
        $\bM^{\ast}$ by Proposition \ref{prop:strict_lyapunov}. In particular, $V$ cannot increase
        along integral curves of $\bF$. Therefore, since $V$ is the squared Euclidian distance to
        $\bM^{\ast}$, every integral curve $t \mapsto \bM(t)$ remains in a compact set. More
        precisely, $\bM([0, \infty)) \subset \{x \in \RR^3 : V(x) \leq V(\bM_{0})\}$, where $\bM_{0}
        \in \RR^{3}$ is the starting point of $\bM$. Because $\bF \in C^1(\RR^3;\RR^3)$, finite-time
        blow-up is therefore impossible, and all solutions are forward complete. The conclusion now
        follows from the global Lyapunov Theorem for strict Lyapunov functions.
    \end{proof}
\end{corollary}

\subsection{Proofs of Theorems in the paper}
\label{appendix:proofs_main_theorems}
In this section we provide the proofs of the theorems presented in the main body of the paper.

\newtheorem*{theoremrestate}{Theorem \ref{thm:admissible_relaxation_fields}}
\begin{theoremrestate}[Characterization of admissible relaxation fields]
    Let $\bF: \RR^{3} \rightarrow \RR^{3}$ be a continuously differentiable vector field. Then $\bF$
    is an admissible relaxation field if and only if there exists $\alpha, \beta \in C^{1}\left( (0,
    \infty) \times \RR \right) \cap C([0, \infty) \times \RR)$, $\zeta \in C^{1}([0, \infty) \times
    \RR)$ such that
	\begin{align*}
		\bF(x) = 
		\begin{pmatrix}
			\alpha(r,x_{3}) & -\beta(r, x_{3}) & 0 \\
			\beta(r, x_{3}) & \alpha(r, x_{3}) & 0 \\
			0 & 0 & 0 
		\end{pmatrix}x + 
		\begin{pmatrix}
			0 \\ 0 \\ \zeta(r, x_{3})
		\end{pmatrix},
	\end{align*}	
	where $r = r(x_{1}, x_{2}) = \sqrt{x_{1}^{2} + x_{2}^{2}}$, and
	\begin{itemize}
	    \item[$(i)$] For all $\tilde x_{3} \in \RR$, we have 
            \begin{align*}
                \lim_{(s, x_{3}) \rightarrow (0, \tilde x_{3})} s \partial_{s} \alpha(s, x_{3}) &= 0, &
                \lim_{(s, x_{3}) \rightarrow (0, \tilde x_{3})}  s \partial_{x_{3}} \alpha(s, x_{3}) &= 0, & \\[2ex]
                \lim_{(s, x_{3}) \rightarrow (0, \tilde x_{3})}  s \partial_{s} \beta(s, x_{3}) &= 0, &
                \lim_{(s, x_{3}) \rightarrow (0, \tilde x_{3})} s \partial_{x_{3}} \beta(s, x_{3})&= 0,
            \end{align*}
            and moreover, $\partial_{s} \zeta(0, \cdot) \equiv 0$.
        \item[$(ii)$]  The only zero of $\zeta(0, \cdot)$ is $M^{\ast}$.
        \item[$(iii)$] $\alpha(0,M^{\ast}) < 0$, $\beta(0, M^{\ast})=0$ and
            $\partial_{x_{3}}\zeta(0,M^{\ast}) < 0$.
        \item[$(iv)$] $s^{2} \alpha(s,x_{3}) + (x_{3}-M^{\ast}) \zeta(s, x_{3}) < 0$, $\forall (s,
            x_{3}) \in \left( [0, \infty) \times \RR \right) \setminus \{(0, M^{\ast})\}$. 
	\end{itemize}
	\begin{proof}
       ``$\Rightarrow$'' Assume $\bF$ is an admissible relaxation field. Then condition
       (\textbf{C1}) implies that there exists $\alpha, \beta \in C^{1}\left( (0,\infty) \times \RR
       \right) \cap C([0, \infty) \times \RR)$ and $\zeta \in C^{1}([0, \infty) \times \RR)$
       satisfying $(i)$, such that $\bF = \bF_{(\alpha, \beta, \zeta)}$, by Proposition
       \ref{prop:characterization_equivariant_vector_fields}. Therefore, since $\bF = \bF_{(\alpha,
       \beta, \zeta)}$, conditions (\textbf{C2}) and (\textbf{C3}) imply $(ii)$ and $(iii)$,
       respectively, by Proposition \ref{prop:local_stable_equilibrium} and Remark
       \ref{remark:real_stable_spectrum}. Similarly, condition (\textbf{C4}) implies $(iv)$ by
       Proposition \ref{prop:strict_lyapunov}. 
	   
       ``$\Leftarrow$'' Conversely, assume $\bF = \bF_{(\alpha, \beta, \zeta)}$ and conditions
       $(i)$--$(iv)$ are satisfied. Then (\textbf{C1}) is satisfied by Proposition
       \ref{prop:characterization_equivariant_vector_fields}. Furthermore, since condition $(iv)$
       implies the inequality in \eqref{eq:alpha_beta_zeta_nonzero}, conditions (\textbf{C2}) and
       (\textbf{C3}) are satisfied by Proposition \ref{prop:local_stable_equilibrium} and Remark
       \ref{remark:real_stable_spectrum}. Finally, (\textbf{C4}) is satisfied by Proposition
       \ref{prop:strict_lyapunov}.
	\end{proof}
\end{theoremrestate}

Next, we provide the proof of Proposition \ref{prop:construction_admissible_fields}. Before doing so, 
however, it will be convenient to have the following intermediate result: 
\begin{lemma}[Differentiability of mappings in $\mathcal{V}_{1}$]
	\label{lemma:s2psi}
    Let $\psi \in \mathcal{V}_{1}$ and define $\Psi: [0, \infty) \times \RR \rightarrow \RR$ by
    $\Psi(s, x_{3}) := s^{2} \psi(s, x_{3})$. Then $\Psi \in \mathcal{V}_{2}$. 
	\begin{proof}
        It is clear that $\Psi \in C^{1}((0, \infty) \times \RR)$, so we only need to analyze the
        behavior of the partial derivatives on $\{0\}\times\RR$. To this end, let $\tilde x_{3} \in
        \RR$ be arbitrary, then 
		\begin{align*}
			\lim_{h \downarrow 0} \frac{ \Psi(h, \tilde x_{3}) - \Psi(0, \tilde x_{3}) }{ h} = 
			\lim_{h \downarrow 0} h \psi(h, \tilde x_{3}) = 0, 
		\end{align*}
        since $\psi$ is continuous at $(0, \tilde x_{3})$. Hence $\partial_{s}\Psi(0, \tilde 
        x_{3}) = 0$. Next, observe that 
		\begin{align*}
            \partial_{s} \Psi(s, x_{3}) = 2s \psi(s, x_{3}) + s^{2} \partial_{s} \psi(s, x_{3}),
            \quad (s, x_{3}) \in (0, \infty) \times \RR. 
		\end{align*} 
        Consequently, since $\psi$ is continuous at $(0, \tilde x_{3})$, and $s \partial_{s} \psi(s,
        x_{3}) \rightarrow 0$ as $(s, x_{3}) \rightarrow (0, \tilde x_{3})$, we see that 
		\begin{align*}
			\lim_{(s, x_{3}) \rightarrow (0, \tilde x_{3})} \partial_{s} \Psi(s, x_{3}) = 0. 
		\end{align*}
		Therefore, $\partial_{s} \Psi$ is continuous at $(0, \tilde x_{3})$. 
		
        It is obvious that $\partial_{x_{3}} \Psi(0, \tilde x_{3})$ exists and must be zero, since
        $\Psi(0, \cdot) \equiv 0$. Furthermore, 
		\begin{align*}
            \partial_{x_{3}} \Psi(s, x_{3}) = s^{2} \partial_{x_{3}} \psi(s, x_{3}), \quad (s,
            x_{3}) \in (0, \infty) \times \RR. 
		\end{align*}
        Therefore, since $s \partial_{x_{3}} \psi(s, x_{3}) \rightarrow 0$ as $(s, x_{3})
        \rightarrow (0, \tilde x_{3})$, we see that
		\begin{align*}
            \lim_{(s, x_{3}) \rightarrow (0, \tilde x_{3})} \partial_{x_{3}} \Psi(s, x_{3}) = 0. 
		\end{align*}
		Hence $\partial_{x_{3}} \Psi$ is continuous at $(0, \tilde x_{3})$. 
		
        Since $\tilde x_{3} \in \RR$ was arbitrary, we may now conclude that $\Psi \in C^{1}([0,
        \infty) \times \RR)$, and moreover $\partial_{s} \Psi(0, \cdot) \equiv 0$, which precisely
        means that $\Psi \in \mathcal{V}_{2}$. 
	\end{proof}
\end{lemma}

\newtheorem*{proprestate}{Proposition \ref{prop:construction_admissible_fields}}
\begin{proprestate}[Construction of admissible relaxation fields]
    Let $a, \beta, \psi \in \mathcal{V}_{1}$ and $c \in \mathcal{V}_{2}$. Assume $a, c >0$,
    $\beta(0, M^{\ast})=0$. Then 
	\begin{align*}
		\bF(x) &= 
		\resizebox{0.85\linewidth}{!}{
		$\begin{pmatrix}
			- a(r, x_{3}) -  (x_{3}-M^{\ast})^{2} \psi(r, x_{3}) & -\beta(r, x_{3}) & 0 \\
			\beta(r, x_{3}) & -a(r, x_{3}) -  (x_{3}-M^{\ast})^{2} \psi(r, x_{3}) & 0 \\
			0 & 0 & r^{2} \psi(r, x_{3}) - c(r, x_{3})
		\end{pmatrix}$} x \nonumber \\[2ex] & \quad + 
		M^{\ast} \begin{pmatrix}
			0 \\ 0 \\ c(r, x_{3}) - r^{2} \psi(r, x_{3})
		\end{pmatrix}.
	\end{align*}	
	is an admissible relaxation field. 	
	\begin{proof}
		First note that $\bF$ is of the form $\bF_{(\alpha, \beta, \zeta)}$ with 
		\begin{align*}
			\alpha(s,x_{3}) = - a(s, x_{3}) -  (x_{3}-M^{\ast})^{2} \psi(s, x_{3}), \quad
			\zeta(s,x_{3}) = (x_{3} - M^{\ast})( s^{2} \psi(s, x_{3}) - c(s, x_{3}) ). 
		\end{align*}
        In particular, $\alpha, \beta \in C^{1}((0, \infty) \times \RR) \cap C([0, \infty) \times
        \RR)$, since $a, \beta, \psi \in \mathcal{V}_{1}$.
        Furthermore, $\zeta \in C^{1}([0, \infty) \times \RR)$, since $\left(s, x_{3}) \mapsto s^{2}
        \psi(s, x_{3}) \right) \in \mathcal{V}_{2}$ by Lemma \ref{lemma:s2psi}, and $c \in
        \mathcal{V}_{2}$ by assumption. Next, we verify properties $(i)$--$(iv)$ of Theorem
        \ref{thm:admissible_relaxation_fields} one by one. 
		
        $(i)$ Observe that $\beta \in \mathcal{V}_{1}$ is already as required. Let $(s, x_{3}) \in
        (0, \infty) \times \RR$ be arbitrary, then 
		\begin{align*}
            \partial_{s} \alpha(s, x_{3}) &= - \partial_{s}a(s, x_{3}) - (x_{3} - M^{\ast})^{2}
            \partial_{s} \psi(s, x_{3}), \\[2ex]
            \partial_{x_{3}} \alpha(s, x_{3}) &= - \partial_{x_{3}}a(s, x_{3}) - (x_{3} - M^{\ast})
            \left( 2 \psi(s, x_{3}) + (x_{3} - M^{\ast}) \partial_{x_{3}} \psi(s, x_{3}) \right).
		\end{align*}
        Since $a, \psi \in \mathcal{V}_{1}$ satisfy the desired limit conditions, and $\psi$ is
        continuous on $\{0\} \times \RR$, we immediately see that 
		\begin{align*}
            \lim_{(s, x_{3}) \rightarrow (0, \tilde x_{3})} s\partial_{s} \alpha(s, x_{3}) = 0,
            \quad \lim_{(s, x_{3}) \rightarrow (0, \tilde x_{3})} s \partial_{x_{3}} \alpha(s,
            x_{3}) = 0, 
		\end{align*}
        for all $\tilde x_{3} \in \RR$. This shows that $\alpha \in \mathcal{V}_{1}$. Similarly, 
		\begin{align*}
            \partial_{s} \zeta(s, x_{3}) &= (x_{3} - M^{\ast}) \left( 2 s \psi(s, x_{3}) + s^{2}
            \partial_{s}\psi(s, x_{3}) - \partial_{s}c(s, x_{3}) \right)
		\end{align*}
		and therefore $\partial_{s} \zeta(0, \cdot) \equiv 0$, since $c \in \mathcal{V}_{2}$.  
		
        $(ii)$ Note that $\zeta(0, x_{3}) = -(x_{3} - M^{\ast}) c(0, x_{3})$. Therefore, since $c
        >0$ by assumption, $\zeta(0, x_{3}) = 0$ if and only if $x_{3} = M^{\ast}$. 
		
        $(iii)$ We have $\beta(0, M^{\ast}) = 0$ by assumption. Furthermore, since $a, c >0$, direct
        computation shows that
		\begin{align*}
            \alpha(0, M^{\ast}) = -a(0, M^{\ast}) < 0, \qquad \partial_{x_{3}} \zeta(0, M^{\ast}) =
            - c(0, M^{\ast}) < 0. 
		\end{align*}
				
		$(iv)$ A straightforward computations shows that 
		\begin{align*}
			s^{2} \alpha(s, x_{3}) + (x_{3} - M^{\ast}) \zeta(s, x_{3}) = 
			- \left( s^{2} a(s, x_{3}) + (x_{3} - M^{\ast})^{2} c(s, x_{3}) \right) < 0, 
		\end{align*}
        for all $(s, x_{3}) \in ([0, \infty) \times \RR) \setminus \{ (0, M^{\ast}) \}$, since $a, c
        >0$ by assumption. 
		
        Altogether, we now conclude that the conditions of Theorem
        \ref{thm:admissible_relaxation_fields} are satisfied and therefore $\bF$ is an admissible
        relaxation field. 
	\end{proof}
\end{proprestate}

\subsection{Tools for constructing admissible relaxation fields}
\label{appendix:construction_admissible_fields}
Proposition \ref{prop:construction_admissible_fields} provides an explicit method to construct
admissible relaxation fields, provided we are able to construct mappings in $\mathcal{V}_{1}$ and
$\mathcal{V}_{2}$. In this section we provide the machinery to construct such mappings. We start by
establishing a straightforward, but slightly different formulation, of the standard result on
the continuity of parameter-dependent integrals.
\begin{lemma}[Continuity of parameter-dependent integrals I]
	\label{lemma:continuous_dependence_integral}
    Let $w \in C((0, \infty) \times \RR)$ and assume that for every $\tilde x_{3} \in \RR$ there
    exists a $\delta>0$ and $\varphi \in L^{1}(0, \delta)$ such that $\vert w(s, x_{3}) \vert \leq
    \varphi(s)$ for a.e. $0<s<\delta$ and $\left \vert x_{3}-\tilde x_{3} \right \vert <\delta$.
    Define $W: (0, \infty) \times \RR \rightarrow \RR$ by 
	\begin{align*}
		W(s, x_{3}) := \int_{0}^{s} w(\rho, x_{3}) \mathrm{d} \rho.
	\end{align*}	
 	Then $W$ is well-defined, and moreover, $W \in C((0, \infty) \times \RR)$. 
	\begin{proof}
        We start by proving that $W$ is well-defined. Let $s >0$ and $\tilde x_{3} \in \RR$ be
        arbitrary, then there exists a $\delta >0$ and $\varphi \in L^{1}(0, \delta)$, both possibly
        depending on $\tilde x_{3}$, such that 
		\begin{align}
			\label{eq:W_well_defined}
			\left \vert w(\rho, \tilde x_{3}) \right \vert \leq 
			\varphi(\rho) \mathbf{1}_{(0,\delta)}(\rho) + 
			\left \vert w(\rho, \tilde x_{3} ) \right \vert \mathbf{1}_{[\delta, s]}(\rho) 
		\end{align}
        for a.e. $\rho \in (0, s]$. Therefore, since the second term on the righthand-side of
        \eqref{eq:W_well_defined} is continuous, and hence integrable, on $[\delta, s]$, we conclude
        that $\rho \mapsto w(\rho, \tilde x_{3})$ is integrable on $(0, s]$. Hence $W$ is
        well-defined, since $s >0$ and $\tilde x_{3} \in \RR$ were arbitrary.
		
        Next, we show that $W$ is continuous on $(0, \infty) \times \RR$. To this end, let $(\tilde
        s, \tilde x_{3}) \in (0, \infty) \times \RR$ be arbitrary. Then for any $s >0$, we have 
		\begin{align}
			\label{eq:W_continuity_1}
			\left \vert W(s, x_{3}) - W(\tilde s, \tilde x_{3}) \right \vert \leq 
		   \int_{0}^{\tilde s} \vert w(\rho, x_{3}) - w(\rho, \tilde x_{3}) \vert \dd \rho + 
           \int_{\min \{s, \tilde s \}}^{\max \{s, \tilde s \}} \left( \vert w(\rho, x_{3}) \vert  +
           \vert w(\rho, \tilde x_{3}) \vert \right) \dd \rho. 
		\end{align}
        Now, let us consider the first term on the righthand-side of \eqref{eq:W_continuity_1}. By
        assumption, there exists a $0 < \delta < \tilde s$ and $\varphi \in L^{1}(0, \delta)$ such
        that 
        \begin{align}
			\label{eq:W_continuity_2}
			 \vert w(\rho, x_{3}) - w(\rho, \tilde x_{3}) \vert &\leq
		    2 \vert \varphi(\rho) \vert \mathbf{1}_{(0, \delta)}(\rho) + 
		   \vert w(\rho, x_{3}) - w(\rho, \tilde x_{3}) \vert  \mathbf{1}_{[\delta, \tilde s]}(\rho) 
		   \nonumber \\[2ex] &\leq
		   2 \vert \varphi(\rho) \vert \mathbf{1}_{(0, \delta)}(\rho) + 
            \sup_{(\rho, z) \in I_{1}} \vert w(\rho, z) - w(\rho, \tilde x_{3}) \vert
            \mathbf{1}_{[\delta, \tilde s]}(\rho) 
		\end{align}
        for a.e. $(\rho, x_{3}) \in (0, \tilde s] \times (\tilde x_{3} - \delta, \tilde x_{3} +
        \delta)$, where $I_{1} := [\delta, \tilde s] \times [\tilde x_{3} - \delta, \tilde x_{3} +
        \delta]$. Since $w$ is continuous on $I_{1}$, and thus bounded, the righthand-side of
        \eqref{eq:W_continuity_2} is integrable on $(0, \tilde s]$. Hence
		 \begin{align*}
            \int_{0}^{\tilde s} \vert w(\rho, x_{3}) - w(\rho, \tilde x_{3}) \vert \dd \rho
            \rightarrow 0
		\end{align*}
		as $x_{3} \rightarrow \tilde x_{3}$ by the Dominated Convergence Theorem. 
		
        Finally, for the second term on the righthand-side of \eqref{eq:W_continuity_1}, we estimate
		\begin{align}	
			\label{eq:W_continuity_3}
            \int_{\min \{s, \tilde s \}}^{\max \{s, \tilde s \}} \left( \vert w(\rho, x_{3}) \vert
            +  \vert w(\rho, \tilde x_{3}) \vert \right) \dd \rho \leq 
		   \left \vert s - \tilde s \right \vert 
           \sup_{(\rho, z) \in I_{2}} \left( \vert w(\rho, z) \vert  +  \vert w(\rho, \tilde x_{3})
           \vert \right),
		\end{align}
        for any $(s, x_{3}) \in \text{int}(I_{2})$, where $I_{2} = \left[ \tilde s -
        \delta, \tilde s + \delta \right ] \times [\tilde x_{3} - \delta, \tilde x_{3} + \delta]$.
        Note that the supremum over $I_{2}$ is a fixed finite number, since $w$ is
        continuous on $I_{2}$. Therefore, \eqref{eq:W_continuity_3} tends to zero as $(s, x_{3})
        \rightarrow (\tilde s, \tilde x_{3})$. Altogether, this shows that $W$ is continuous at
        $(\tilde s, \tilde x_{3}) \in (0, \infty) \times \RR$. Therefore, since $(\tilde s, \tilde
        x_{3}) \in (0, \infty) \times \RR$ was arbitrary, we conclude that $W \in C((0, \infty)
        \times \RR)$. 
	\end{proof}
\end{lemma}

The standard result on the continuity of parameter-dependent integrals can be seen as a 
direct consequence of Lemma \ref{lemma:continuous_dependence_integral}, as shown in the next 
corollary. 
\begin{corollary}[Continuity of parameter-dependent integrals II]
	\label{corollary:continuous_dependence_integral_2}
    Let $w \in C([0, \infty) \times \RR)$ and define $W: [0, \infty) \times \RR \rightarrow \RR$ by 
	\begin{align*}
		W(s, x_{3}) := \int_{0}^{s} w(\rho, x_{3}) \mathrm{d} \rho.
	\end{align*}	
 	Then $W \in C([0, \infty) \times \RR)$. 
	\begin{proof}
        First note that for any $x_{3} \in \RR$ and $\delta >0$, the map $w$ is bounded on $I(x_{3},
        \delta) := [0, \delta] \times [x_{3} - \delta, x_{3} + \delta]$, since $w \in C([0, \infty)
        \times \RR)$. Therefore, the hypothesis of Lemma \ref{lemma:continuous_dependence_integral}
        is trivially satisfied, and hence $W \in  C((0, \infty) \times \RR)$. To establish
        continuity on $\{0\} \times \RR$, let $\tilde x_{3} \in \RR$ be arbitrary and observe that 
		\begin{align}
			\label{eq:W_continuity_zero}
			\left \vert W(s, x_{3}) - W(0, \tilde x_{3}) \right \vert \leq 
			\int_{0}^{s} \left \vert w(\rho, x_{3}) \right \vert \dd \rho \leq
			s \sup_{(\rho, z) \in I(\tilde x_{3}, \delta)} \vert w(\rho, z) \vert
		\end{align}
        for all $0 \leq s < \delta$ and $x_{3} \in (\tilde x_{3} - \delta, \tilde x_{3} + \delta)$.
        The supremum over $I(\tilde x_{3}, \delta)$ is a fixed finite number, since $w$ is
        continuous on $I(\tilde x_{3}, \delta)$. Hence the righthand-side of
        \eqref{eq:W_continuity_zero} converges to zero as $(s, x_{3}) \rightarrow (0, \tilde
        x_{3})$. Since $\tilde x_{3} \in \RR$ was arbitrary, we conclude that $W$ is continuous on
        $\{0\} \times \RR$, thus proving that $W \in  C([0, \infty) \times \RR)$.
	\end{proof}
\end{corollary}

We are now ready to explain how to explicitly constructing mappings in $\mathcal{V}_{1}$ and
$\mathcal{V}_{2}$. 
\begin{lemma}[Construction of mappings in $\mathcal{V}_{1}$]
    \label{lemma:construction_V1}
    Let $f \in C^{1}(\RR)$ and assume $g \in C([0, \infty) \times \RR)$ satisfies the following
    conditions:
	\begin{itemize}
	    \item[$(i)$] For every $\tilde x_{3} \in \RR$ there exists a $\delta_{1}>0$ and 
            $\eta \in L^{1}(0,\delta_{1})$ such that
            \begin{align*}
                 \left \vert \frac{g(s,x_3)}{s} \right \vert \leq \eta(s) 
            \end{align*}	
            for a.e. $0<s<\delta_{1}$ and $\left \vert x_{3}-\tilde x_{3} \right \vert <\delta_{1}$. 
        \item[$(ii)$] $\partial_{x_{3}}g$ exists on $(0, \infty) \times \RR$ and is continuous.
        \item[$(iii)$] For every $\tilde x_{3} \in \RR$ there exists a $\delta_{2}>0$ and 
            $\upsilon \in L^{1}(0,\delta_{2})$ such that
            \begin{align*}
                \left \vert \frac{\partial_{x_3}g(s,x_3)}{s} \right \vert \leq \upsilon(s)
            \end{align*}
         for a.e. $0<s<\delta_{2}$ and $\left \vert x_{3}-\tilde x_{3} \right \vert <\delta_{2}$.         
	    \item[$(iv)$] $g(0, \cdot) \equiv 0$. 
	\end{itemize} 
	Then the map $G: [0, \infty) \times \RR \rightarrow \RR$ defined by
	\begin{align}
		\label{eq:explicit_map_V1}
		G(s, x_{3}) = 
		\begin{cases}
            f(x_{3}) + \displaystyle \int_{0}^{s} \frac{ g(\rho, x_{3} ) }{ \rho} \mathrm{d} \rho, & (s,
            x_{3}) \in (0, \infty) \times \RR \\
			f(x_{3}), & s=0, \ x_{3} \in \RR
		\end{cases}
	\end{align}
	is well-defined, and moreover, $G \in \mathcal{V}_{1}$. 
	\begin{proof}
         First note that $G$ is well-defined by Lemma \ref{lemma:continuous_dependence_integral},
         since $(\rho, x_{3}) \mapsto g(\rho, x_{3})/\rho$ is continuous on $(0, \infty) \times
         \RR$, and satisfies $(i)$. Next, we show that the partial derivatives of $G$ exist and are
         continuous on $(0, \infty) \times \RR$. It is a straightforward consequence of the
         Fundamental Theorem of Calculus that $\partial_{s}G$ exists and 
		 \begin{align}
		 	\label{eq:DsG}
            \partial_{s}G(s, x_{3}) = \frac{g(s, x_{3})}{s}, \quad \forall (s, x_{3}) \in (0,
            \infty) \times \RR. 
		 \end{align}
         In particular, $\partial_{s}G \in C((0, \infty) \times \RR)$, since $g \in C([0, \infty)
         \times \RR)$. 
		 
         To analyze $\partial_{x_{3}}G$, let $s >0$ and $\tilde x_{3} \in \RR$ be arbitrary. Note
         that $x_{3} \mapsto g(\rho, x_{3}) / \rho$ is continuously differentiable on $\RR$, for
         each fixed $\rho >0$, by $(ii)$. Therefore, for any $\rho \in (0,s]$ and $h \in \RR$, there
         exists $\xi(h, \rho, \tilde x_{3}) \in (\tilde x_{3} -\vert h \vert, \tilde x_{3} + \vert h
         \vert)$ such that 
		 \begin{align*}
		 	\frac{ g(\rho, \tilde x_{3} + h) - g(\rho, \tilde x_{3}) }{ h \rho} = 
			 \frac{ \partial_{x_{3}}g(\rho, \xi(h, \rho, \tilde x_{3}))}{\rho} 
		 \end{align*}
         by the Mean Value Theorem. Choose $0 < \delta_{2} < s$ and $\upsilon \in L^{1}(0,
         \delta_{2})$ as in $(iii)$, and note that for $\vert h \vert < \delta_{2}$, we must have
         $\xi(h, \rho, \tilde x_{3}) \in (\tilde x_{3} - \delta_{2}, \tilde x_{3} + \delta_{2})$.
         Hence
		 \begin{align}
		 	\label{eq:Dx3G_existence}
            \left \vert \frac{ g(\rho, \tilde x_{3} + h) - g(\rho, \tilde x_{3}) }{ h \rho} \right
            \vert &\leq 
			\vert \upsilon(\rho) \vert \mathbf{1}_{(0, \delta_{2})}(\rho) + 
            \left \vert  \frac{ \partial_{x_{3}}g(\rho, \xi(h, \rho, \tilde x_{3}))}{\rho} \right
            \vert \mathbf{1}_{[\delta_{2}, s]}(\rho) 
			\nonumber \\[2ex] & \leq
			 \vert \upsilon(\rho) \vert \mathbf{1}_{(0, \delta_{2})}(\rho) +
             \sup_{(\rho, z) \in I} \left \vert \frac{ \partial_{x_{3}}g(\rho, z) }{ \rho } \right
             \vert \mathbf{1}_{[\delta_{2}, s]}(\rho)
		 \end{align}
         for all $\vert h \vert < \delta_{2}$ and $\rho \in (0, s]$, where $I := [\delta_{2}, s]
         \times [\tilde x_{3} - \delta_{2}, \tilde x_{3} + \delta_{2}]$. Since $\partial_{x_{3}}g$
         is continuous on $I$ by $(ii)$, and hence bounded, the righthand-side of
         \eqref{eq:Dx3G_existence} is integrable on $(0, s]$. Therefore, it follows from the
         Dominated Convergence Theorem that $\partial_{x_{3}}G(s, \tilde x_{3})$ exists, and
         moreover, 
		 \begin{align}
		 	\label{eq:Dx3G}
            \partial_{x_{3}}G(s, \tilde x_{3}) = f'(\tilde x_{3}) + \int_{0}^{s} \frac{
            \partial_{x_{3}}g(\rho, \tilde x_{3})}{ \rho} \dd \rho. 
		 \end{align}
         Furthermore, since $f' \in C(\RR)$, and $(\rho, x_{3}) \rightarrow \partial_{x_{3}}g(\rho,
         x_{3})/\rho$ is continuous on $(0, \infty) \times \RR$ and satisfies $(iii)$, we conclude
         that $\partial_{x_{3}}G \in C((0, \infty) \times \RR)$ by Lemma
         \ref{lemma:continuous_dependence_integral}. Hence $G \in C^{1}((0, \infty) \times \RR)$.  
		 
         Next, we show that the partial derivatives of $G$ satisfy the desired limits. It follows
         directly from \eqref{eq:DsG} that
		 \begin{align*}
            \lim_{(s, x_{3}) \rightarrow (0, \tilde x_{3})} s\partial_{s}G(s, x_{3}) = g(0, \tilde
            x_{3}) = 0, 
		 \end{align*}
         since $g$ is continuous at $(0, \tilde x_{3})$ and by $(iv)$. Furthermore, using $(iii)$
         again, we see that 
		 \begin{align*}
            \left \vert \int_{0}^{s} \frac{ \partial_{x_{3}}g(\rho, x_{3})}{ \rho} \dd \rho \right
            \vert \leq \int_{0}^{s} \upsilon(\rho) \dd \rho
		 \end{align*}
         for all $0 < s < \delta_{2}$ and $x_{3} \in (\tilde x_{3} - \delta_{2}, \tilde x_{3} +
         \delta_{2})$. The righthand-side of the latter inequality converges to $0$ as  $(s, x_{3})
         \rightarrow (0, \tilde x_{3})$, since $\upsilon \in L^{1}(0, \delta_{2})$. Consequently,
		 \begin{align*}
		   \lim_{(s, x_{3}) \rightarrow (0, \tilde x_{3})} s \partial_{x_{3}}G(s, x_{3}) = 0, 
		 \end{align*}
		 where we also used the continuity of $f'$ at $\tilde x_{3}$. 
		  
         It remains to show that $G$ is continuous on $\{0\} \times \RR$. Using $(i)$ again, we see
         that 
		 \begin{align*}
		 	\left \vert G(s, x_{3}) - G(0, \tilde x_{3}) \right \vert \leq
		  \vert f(x_{3}) - f(\tilde x_{3}) \vert + 
		  \int_{0}^{s} \left \vert \frac{ g(\rho, x_{3}) }{ \rho} \right \vert \dd \rho \leq 
		  \vert f(x_{3}) - f(\tilde x_{3}) \vert + 
		  \int_{0}^{s} \eta(\rho) \dd \rho
		 \end{align*}	
         for all $0 < s < \delta_{1}$ and $x_{3} \in (\tilde x_{3} - \delta_{1}, \tilde x_{3} +
         \delta_{1})$. The righthand-side of the latter inequality converges to $0$ as $(s, x_{3})
         \rightarrow (0, \tilde x_{3})$, since $f$ is continuous at $\tilde x_{3}$ and $\eta \in
         L^{1}(0, \delta_{1})$, thus proving that $G$ is continuous at $(0, \tilde x_{3})$.
         Altogether, we now conclude that $G \in C^{1}((0, \infty) \times \RR) \cap C([0, \infty)
         \times \RR)$, and satisfies the desired limits. Hence $G \in \mathcal{V}_{1}$. 
	\end{proof}
\end{lemma}

\begin{lemma}[Construction of mappings in $\mathcal{V}_{2}$]
    \label{lemma:construction_V2}
    Let $f \in C^{1}(\RR)$ and assume $g \in C([0, \infty) \times \RR)$ satisfies the following
    conditions:
    \begin{itemize}
    	\item [$(i)$] $g(0, \cdot) \equiv 0$.
	\item[$(ii)$] $\partial_{x_{3}}g$ exists on $[0, \infty) \times \RR$ and is continuous.
    \end{itemize}
    Define $G: [0, \infty) \times \RR \rightarrow \RR$ by 
	\begin{align*}
		G(s,x_{3}) := f(x_{3}) + \int_{0}^{s} g(\rho, x_{3}) \mathrm{d} \rho.
	\end{align*}
    Then $G \in \mathcal{V}_{2}$. 
    \begin{proof}
        We first analyze $\partial_{s}G$. It follows directly from the 
        Fundamental Theorem of Calculus that 
        \begin{align*}
                \partial_{s}G(s, x_{3}) = g(s, x_{3}), \quad \forall (s, x_{3}) \in (0, \infty)
                \times \RR,
        \end{align*} 
        since $g(\cdot, x_{3}) \in C([0, \infty))$ for all $x_{3} \in \RR$. Furthermore, for $s=0$,
        we have 
        \begin{align*}
            \lim_{h \downarrow 0} \frac{ G(h, x_{3}) - G(0, x_{3}) }{h} = 
            \lim_{h \downarrow 0} \frac{1}{h} \int_{0}^{h} g(\rho, x_{3}) \dd \rho = 
            g(0, x_{3}), 
        \end{align*}
        since $g$ is continuous at $(0, x_{3})$, for all $x_{3} \in \RR$. Hence $\partial_{s}G = g$
        and therefore $\partial_{s}G \in C([0, \infty) \times \RR)$. Moreover, $\partial_{s}G(0,
        \cdot) = g(0,\cdot) \equiv 0$ by $(i)$. 
        
        Next, we consider $\partial_{x_{3}}G$. For fixed $s \geq 0$, differentiating under the integral 
        sign with respect to $x_{3}$ is justified by $(ii)$ and the Dominated Convergence Theorem: 
        the difference quotients (in the variable $x_{3}$) converge pointwise to $\partial_{x_{3}}
        g(\rho, x_{3})$, and by the mean value theorem are bounded by the supremum of 
        $\vert \partial_{x_{3}}g \vert$ on a compact neighborhood of the form $[0, s] \times 
        [x_{3} - \delta, x_{3} + \delta]$. Since $\partial_{x_{3}}g$ is continuous, this bound is 
        finite. Therefore, since $f \in C^{1}(\RR)$, the partial derivative $\partial_{x_{3}}G$ 
        exists and
	    \begin{align}
		\label{eq:d3G}
            \partial_{x_{3}}G(s, x_{3}) = f'(x_{3}) + \int_{0}^{s} \partial_{x_{3}}g(\rho, x_{3}) \dd
            \rho, \quad \forall (s, x_{3}) \in [0, \infty) \times \RR. 
	      \end{align}
        Moreover, since $f' \in C(\RR)$ and $\partial_{x_{3}}g \in C([0, \infty) \times \RR)$, it
        follows from Corollary \ref{corollary:continuous_dependence_integral_2} that
        $\partial_{x_{3}}G \in  C([0, \infty) \times \RR)$. Altogether, this proves that $G \in
        \mathcal{V}_{2}$.
    \end{proof}
\end{lemma}

The representation in Lemma~\ref{lemma:construction_V2} captures a large set of functions in
$\mathcal{V}_{2}$, but not all of them. The only additional restriction is that we have
inadvertently imposed extra regularity, as \eqref{eq:d3G} is differentiable in $s$. Thus the loss of
generality amounts solely to an additional regularity assumption in the longitudinal direction.

\bibliographystyle{unsrt}
\bibliography{sample}

\end{document}